\documentclass[11pt,a4paper]{article}
\usepackage[utf8]{inputenc}
\usepackage{amsmath}
\usepackage{amsfonts}
\usepackage{amssymb}
\usepackage{amsthm} 
\usepackage{graphicx}
\usepackage[table, usenames, dvipsnames]{xcolor}
\usepackage[all]{xy}
\usepackage{placeins} 
\usepackage{tikz} 
\usepackage[version=3]{mhchem} 
\usepackage{url}

\usepackage[breaklinks]{hyperref}
\usepackage{doi}
\usepackage{settobox} 
\newbox{\subbox}
\newlength{\xheight}
\newlength{\subheight}
\newcommand\xheightsub[2]{%
	\savebox{\subbox}{\({}_{#2}\)}%
	\settoheight{\xheight}{\({}_x\)}%
	\settoboxheight{\subheight}{\subbox}%
	\addtolength{\subheight}{-\xheight}%
	{#1}{\raisebox{-\subheight}{\usebox{\subbox}}}%
}
\usepackage[margin=3cm]{geometry}
\usepackage[sort&compress,comma,square,numbers]{natbib} 

\theoremstyle{plain}
\newtheorem{thm}{Theorem}[section]
\newtheorem{proposition}[thm]{Proposition}
\newtheorem{corollary}[thm]{Corollary}
\newtheorem{lem}[thm]{Lemma}

\theoremstyle{definition}
\newtheorem{definition}[thm]{Definition}
\newtheorem{remark}[thm]{Remark}
\newtheorem{example}[thm]{Example}

\usepackage{multirow}

\newcommand{\cN}{\mathcal{N}}
\newcommand{\cS}{\mathcal{S}}
\newcommand{\cC}{\mathcal{C}}
\newcommand{\cR}{\mathcal{R}}
\newcommand{\R}{\mathbb{R}}
\newcommand{\N}{\mathbb{N}}
\newcommand{\sgn}{{\rm sgn}}
\newcommand{\fr}{\mathfrak{r}}

\newcommand{\footremember}[2]{%
	\footnote{#2}
	\newcounter{#1}
	\setcounter{#1}{\value{footnote}}%
}
\newcommand{\footrecall}[1]{%
	\footnotemark[\value{#1}]%
} 

\author{Shenghao Yao\footremember{this}{Research Centre for Computational Sciences and Mathematical Modelling, Coventry University, Coventry, United Kingdom.}
	, AmirHosein Sadeghimanesh\footrecall{this} \footremember{Contact-AS}{amirhossein.sadeghimanesh@coventry.ac.uk.}\footremember{corr}{Corresponding author.} and Matthew England\footrecall{this} 
	}

\title{Understanding Multistationarity of Fully Open Reaction Networks}

\begin{document}

\maketitle

\begin{abstract}
This work addresses multistationarity of fully open reaction networks equipped with mass action kinetics. We improve upon existing results relating existence of positive feedback loops in a reaction network and multistationarity; and we provide a novel deterministic operation to generate new non-multistationary networks. This is interesting because while there were many operations to create infinitely many new multistationary networks from a multistationary example, this is the first such operation for the non-multistationary counterpart. 

Such tools for the generation of example networks have a use-case in the application of data science to reaction network theory.  We demonstrate this by using new data, along with a novel graph representation of reaction networks that is unique up to a permutation on the name of species of the network, to train a graph attention neural network model to predict multistationarity of reaction networks. This is the first time machine learning tools are used for studying classification problems of reaction networks.
	
\end{abstract}

\section{Introduction}
\label{sec:introduction}

\subsection{From the life sciences to CRNs}
\label{sec:introduction_subsec_1}

Many phenomena in the life sciences involve a group of species, such as chemical molecules, proteins or cells, depending on the context; and some interactions among them. Each interaction causes some of these species to be destroyed / consumed and some others to be created / produced. Thus the quantities of these species, which can be a concentration (a real number) or a count (an integer), are changing over time. This means the quantity of each species is a variable, i.e. function of time. These changes over time, caused by interactions, give rise to derivatives of the variables and equations describing these changes known as ordinary differential equations (ODEs). The set of these equations defines an ODE system. Note that in this paper we are considering continuous modelling; in the case of discrete modelling, one enters the realm of stochastic processes. 

To predict the behaviour of experiments in chemistry or biology, one can study their corresponding ODE system mathematically. For example, consider identifying a stable equilibrium of the experiment. Stable equilibria are reached after waiting for sufficient time for the concentrations of the species to not vary any more. By solving the algebraic equations obtained by setting the derivatives in the ODE system to zero we may identify the values of concentrations at equilibria, both the stable and unstable ones known together as steady states.

The field of study devoted to developing the mathematics which describes the changes in the concentrations of a group of species with respect to the interactions among them is called \emph{Chemical Reaction Network (CRN) Theory} \cite{Feinberg-2019}. We note that despite the explicit mention of the word chemical in the title of this field, the concept and the results are equally applicable to many other domains, including but not restricted to biology, epidemiology, ecology, population dynamics, and pharmacokinetics. In this paper we maintain a common language, with the interactions among the species in any such system called the \emph{reactions}, and the system itself called a \emph{reaction network}. The network is \emph{fully open} if it has inflow and outflow reactions for all species.

Some systems do not reach any steady states, such as chaotic or oscillatory models: for example, the classic non-equilibrium oscillating Belousov-Zhabotinsky reaction in chemistry \cite{Hudson-Mankin-1981}, the non-equilibrium thermodynamic system described by the Boltzmann model \cite{Diperna-Lions-1989}, and the Lotka-Volterra equations that describe the dynamic interaction between predator and prey in biology \cite{Abakuks-1982}. Some systems always have a unique steady state: for instance, the weakly reversible complex balanced reaction networks have one steady state in a closed system \cite{Horn-Jackson-1972}. Finally, some models can have more than one steady state: such models are called \emph{multistationary}. Multistationarity allows for the possibility of more than one choice in the system and is a key property in allowing mechanisms such as switch-like behaviour in cell-division or apoptosis decisions\cite{Guantes-Poyatos-2008,Sen-Bagci-Camurdan-2011}. It also enables the role of memory in cells and bacteria \cite{Kothamachu-Feliu-Cardelli-Soyer-2015}, and is crucial in designing biological or chemical circuits \cite{Thomas-Kaufman-2001,Gardner-Cantor-Collins-2000}.  Detecting multistationarity in networks is thus of great interest.

\subsection{Detecting multistationary}
\label{sec:introduction_subsec_multi}

The most well-known set of algorithms to check if a reaction network is multistationary are \emph{deficiency algorithms}. This class of tools was initially developed in the group of Feinberg \cite{Feinberg-1972,Feinberg-1987,Feinberg-1995,Ellison-1998}. Such algorithms work on a specific type of reaction networks called regular networks. While regularity may sound a natural assumption in some application domains, it does not hold for reaction networks in many other areas.

An alternative set of algorithms, which in theory work for any arbitrary reaction network, use algebraic tools \cite{Buchberger2006, Collins1998b} to find the exact semi-algebraic description of the parameter values for which the system has multiple steady states, and thus where the network is multistationary  \cite{Lichtblau-2021,Rost-Amir-2021,Rost-Amir-2023,Bradford-Davenport-England-Errami-et-al-2020}. 
Unfortunately, these algebraic algorithms have worst case computational complexity doubly exponential in the number of species and parameters of the network \cite{Bradford-Davenport-England-McCallum-Wilson-2016,Mayr-Meyer-1982,Mayr-Ritscher-2013}.  Thus, when these number more than just a few, the computation does not terminate on a computer in reasonable time; or more commonly the intermediate algebraic expression swell exhausts the memory resources \cite{Rost-Amir-2021}.

To overcome this challenge it is natural to look for patterns in smaller pieces of the network that can inform us about the multistationarity of the whole network, instead of studying the whole network directly which is computationally expensive. One such approach links the possibility of being multistationary to the existence of positive feedback loops (see Section~\ref{sec:multi_and_feedback_loops}).  The first contribution of the present paper is to improve upon the previous results making this link for the case of fully open reaction networks equipped with mass action kinetics (see Proposition~\ref{prop:positive_feedback_loops_for_fully_open}). We then use this to make our second contribution: a deterministic method to generate infinitely many new non-multistationary networks from a non-multistationary example, see Proposition~\ref{prop:changing_sc_without_loosing_NPFL}. Previous similar operations in the literature, such as extending a network by adding species or reactions, could only produce multistationary networks from a multistationary example. Our work is the first such operation for the non-multistationary counterpart.  One motivation for operations to generate networks with such properties is to build datasets that allow for the application of data science and Machine Learning (ML) to reaction network theory. 

\subsection{ML for multistationarity in CRNs}
\label{sec:introduction_subset_ML}

A problem of significant interest is to decide whether a given network is multistationary or not. This is a binary classification problem: thinking of the networks as input data that we want to classify as either ``\textit{multistationary}" or ``\textit{not multistationary}". A common approach for classification problems is to use ML techniques. 

ML tools have competed with, and enhanced, traditional methods in many research areas over recent years. An advantage of using ML is that we pay the heavy computational time cost once at the start, during the training process; then once we have a trained model, we can get predictions of the label for new input data from it almost instantaneously, for as many inputs as we need.  
The use of ML in CRN theory up to now has been limited to predicting species, complexes and reactions involved in a reaction network rather than predicting behaviour of the network, such as learning the potential energy surfaces and reaction pathways. For a comprehensive review on the current applications of ML in CRN theory see \cite{Wen-et-al-2023}. 

Our third contribution is introducing a new graph representation for reaction networks (Definition~\ref{def:SRSC_graph}) suitable for machine learning. A reaction network can be uniquely determined from this graph up to a permutation in the name of species, which does not affect whether the network is multistationary. This representation of reaction networks allows us to cope with the variable length of CRN data, and make use of graph learning algorithms. Our final contribution is generating the first large scale dataset of reaction networks labelled according to whether they exhibit multistationary (using several tools including those developed in this paper), and to demonstrate its use in training an ML classifier of multistationarity.  

The methodology developed will not only let us to address the multistationarity classification problem in fully open CRN systems with the power of ML, but also opens the path for others to use ML tools for other questions in CRN theory. 

We note that this paper focusses its attention on fully open networks.  Our initial motivation for this work was Theorem~\ref{thm:one_non_flow_Joshi_thm} (\cite[Theorem 4.1]{Joshi-2013}) where the multistationarity of a fully open network with one non-flow reaction can be decided by a simple inequality on the stoichiometric coefficients of the non-flow reaction and its direction.  Since there already exists such theoretical results providing simple rules to classify multistationarity for the fully open case, it seems like a reasonable place to focus our initial efforts.
Another reason for this focus is the requirements in the results we used to create our labelled dataset. It is certainly desirable to extend this work to create a more generic labelled dataset, but that will first require significant extra work on new mathematical understanding for such networks. We observe that this may exemplify a new working relationship between mathematicians and ML: where increasingly the ML tools may take over the work of analysing systems; and mathematicians turn their attention away from designing such algorithms directly, to instead focus on finding results that create better datasets for the ML to learn from.

\subsection{Structure of the paper}
\label{sec:paper_structure}

We continue in Section~\ref{sec:CRN_definition} with a formal definition of CRNs, and then detail approaches to detect multistationarity of a CRN in  Section~\ref{sec:multistationarity_methods}, some from the literature and some new results. We briefly discuss the limitation of applying simple ML tools such as support vector machine and random forests in Section \ref{sec:ML}, explaining that the variable length input of the problem requires more advanced ML tools. In Section~\ref{sec:GAT} we present a new graph representation for CRNs, generate a large labelled dataset (making use of the methods in Section \ref{sec:multistationarity_methods}), and present results for a graph learning algorithm (a graph attention network) to predict multistationarity. We finish with ideas for future work in Section~\ref{sec:conclusions}.

\subsection{Notation}
\label{sec:notations}

The cardinality of a set, $A$, is denoted by $|A|$. We assume that $0^0=1$, and so there is no ambiguity for the notation $x^0$ where $x$ is a real-valued variable. For any real number $r\in\R$, the sign of $r$ is denoted by $\sgn(r)$ and defined to be $1$, $-1$ or $0$ if $r$ is positive, negative or zero respectively.

\section{Chemical reaction network theory}
\label{sec:CRN_definition}

\subsection{Chemical reaction networks}
\label{sec:CRN_definition_subsec}

A \emph{chemical reaction network}, or a \emph{network} for short, is denoted by $\cN$ and  consists of three finite sets: the set of species, $\cS$, the set of complexes, $\cC$, and the set of reactions, $\cR$. Every species of $\cS$ needs to take part in at least one complex from $\cC$, and every complex of $\cC$ needs to take part in at least one reaction from $\cR$.  
The elements of $\cS$ are denoted by uppercase letters, i.e. $A$, $B$, $\dots$ or $X_1$, $X_2$, $\dots$ and an element of $\cC$ is a linear combination of species with non-negative integer coefficients. For example if $\cS=\lbrace A, B, C \rbrace$, then $A+2C$ is a complex (with the coefficient of $B$ in this particular complex set to zero). The coefficient beside each species in a complex is called a \emph{stoichiometric coefficient}. An element of $\cR$ is an ordered pair of complexes. If $y$ and $y'$ are two complexes, then the reaction defined by the ordered pair $(y, y')$ is denoted $y\ce{->}y'$. The complex on the left side of the reaction arrow is called the \emph{reactant} and the complex on the right side is called the \emph{product}. A reversible reaction is denoted as $y\ce{<=>}y'$ which is simply a shorthand for two irreversible reactions $y\ce{->}y'$ and $y'\ce{->}y$. 

The \emph{zero complex} is the complex in which the stoichiometric coefficients of all species are zero.  A \emph{flow reaction} has the zero complex as either the reactant (an \emph{inflow reaction}) or the product (an \emph{outflow reaction}).  An inflow reaction denotes the injection of species to the network and an outflow denotes the degradation of some species or its extraction from the network.

\begin{definition}\label{def:fully_open_network}
	We say that a reaction network $\cN=(\cS,\cC,\cR)$ is \emph{fully open} if 
	\[\forall X\in\cS\;\colon\;0\ce{->}X,X\ce{->}0\in\cR.\]
	The reactions $0\ce{->}X$ and $X\ce{->}0$ are called the \emph{inflow} and the \emph{outflow} of the species $X$.
\end{definition}

A directed graph can naturally be associated to a network by considering the complexes as its nodes and the reactions as its directed edges. We refer to this digraph as the \emph{CR-graph} of the network.

\begin{example}[\cite{Drexler-et-al-2019}]
	\label{ex:chemotherapy}
	Consider a simple model for chemotherapy of tumours. Let $A$ stand for tumour cells and $B$ for a drug. In a simple scenario a tumour cell divides and becomes two new tumour cells, giving reaction $A\ce{->}2A$.  Suppose the drug and a tumour cell react and the tumour cell dies. This can be expressed as a reaction with one molecule of the drug and one cell in the reactant, but only the molecule of the drug remaining in the product, i.e. $A+B\ce{->}B$. Finally, we assume the drug leaves the body so we have an outflow for $B$, $B\ce{->}0$. The three sets associated with this network are $\cS=\lbrace A, B\rbrace$, $\cC=\lbrace 0, A, 2A, B, A+B\rbrace$ and 
	\[\cR=\lbrace A\ce{->}2A, A+B\ce{->}B, B\ce{->}0\rbrace.\]
	This information can be represented via the following CR-graph with two components.
	\[A\ce{->}2A \qquad A+B\ce{->}B\ce{->}0\]
	Since there are no flow reactions for $A$ and no inflow for $B$, this is not a fully open network (Definition~\ref{def:fully_open_network}).
\end{example}

\subsection{ODEs associated to CRNs}
\label{sec:ODE_system}

The amount of each species changes over time and thus we associate a variable (technically a mathematical function with respect to time) to each species, denoted with lowercase letters. For example the amount of species $X$ at time $t$ should be denoted by $x(t)$. However, we drop the emphasis on $t$ and simply write $x$ to simplify the notation. Similarly, instead of $\frac{dx(t)}{dt}$ for the derivative of $x$ with respect to $t$, we use $\dot{x}$.

Let $\cS=\lbrace X_1,\dots,X_n\rbrace$ and consider a complex $y=\sum_{i=1}^n y_iX_i$.  We use the notation $y$ to represent both the complex and the vector of the stoichiometric coefficients $(y_1,\dots,y_n)$ interchangeably. Each time a reaction $y\ce{->}y'$ occurs, the amount of $X_i$ changes by $y'_i-y_i$ units. If the rate of occurrence of the reaction $y\ce{->}y'$ is denoted by $\rho_{y\to y'}$, then the total change in $x_i$ can be formulated as
\begin{equation}\label{eq:ODE_equation_1}
	\dot{x}_i=\sum_{y \to y'\in\cR}(y'_i-y_i)\rho_{y \to y'}.
\end{equation}

The rate of occurrence of a reaction depends on the choice of the kinetics (our assumption about how often reactions happen). The most common kinetics, which we use in this paper, is \emph{mass action kinetics} which states that the rate of occurrence of a reaction is proportional to the product of the concentration of the species on the reactant side of the reaction. 
The constant of the proportionality is called the \emph{reaction rate constant} and is denoted by $k_{y \to y'}$. Therefore $\rho_{y \to y'}=k_{y \to y'}\prod_{i=1}^n x_i^{y_i}$. 
We simply use $x^y$ to denote  $\prod_{i=1}^n x_i^{y_i}$. Then Equation~\eqref{eq:ODE_equation_1} becomes $\dot{x}_i=\sum_{y \to y' \in \cR}(y_i'-y_i)k_{y \to y'}x^y$. Using a vector notation and putting the equations for all $n$ species together, we get the following system of ordinary differential equations (ODEs):
\begin{equation}\label{eq:ODE_system}
	\dot{x}=\sum_{y \to y'\in\cR}(y'-y)k_{y \to y'}x^y.
\end{equation}

The reaction rate constants may be included in the CR-graph as labels for the corresponding directed edges.

\begin{example}
	\label{ex:chemotherapy_ODE}
	Recall the network in Example~\ref{ex:chemotherapy}. Rename the species as $X_1=A$ and $X_2=B$. Then order the reactions in the order they were above: thus instead of using $k_{A \to 2A}$ we use $k_1$. The labelled CR-graph becomes the following.
	\[X_1\ce{->[k_1]}2X_1 \qquad X_1+X_2\ce{->[k_2]}X_2\ce{->[k_3]}0.\]
	The ODE system associated to this network is
	\begin{equation}\label{eq:chemotherapy_ODE}
		\left\lbrace\begin{array}{l}
			\dot{x}_1 = k_1x_1-k_2x_1x_2\\
			\dot{x}_2 = -k_3x_2
		\end{array}\right..
	\end{equation}
\end{example}

Note that the reaction rate constants, the $k_i$s, are parameters for the ODE system and can only attain positive real values. On the other hand, the $x_i$s are variables and can have non-negative real values.

\begin{definition}\label{def:steady_states}
	Let $\cN$ be a network with $\cS=\lbrace X_1,\dots,X_n\rbrace$. For a fixed value of the reaction rate constants, we say $x^\star=(x_1^\star,\dots,x_n^\star)\in\R_{\geq 0}^n$ is a \emph{steady state} if after substituting $x=x^\star$ in Equation~\eqref{eq:ODE_system} we get $\dot{x}=0$. The polynomials on the right hand side of Equation~\eqref{eq:ODE_system} are called the \emph{steady state polynomials}, with \emph{steady state equations} being those equations obtained by setting the steady state polynomials to zero.
\end{definition}

In the case of the simple chemotherapy network in Examples~\ref{ex:chemotherapy} and \ref{ex:chemotherapy_ODE} for any choice of parameter values $k\in\R_{>0}^3$, it is readily seen from the steady state equations that the only steady state is $x=(0,0)$. However this steady state can be attained only in the case where $x_1(0)=0$. The steady state $(0,0)$ is an unstable steady state because any generic solution, no matter how close it gets to it, eventually moves away from it ($x_1$ starts diverging to infinity). Although stability of steady states is important, we do not focus on stability in the rest of this paper.

\subsection{Multistationarity}
\label{sec:multistationarity}

\begin{definition}\label{def:multistationarity}
	A fully open network with $n$ species and $r$ reactions is called \emph{multistationary} if there exists a choice of reaction rate constants, $k^\star\in\R_{>0}^r$, for which the system of steady state equations after substituting $k=k^\star$ has more than one solution in $\R_{> 0}^n$.
\end{definition}

\begin{remark}\label{rem:fully_open_no_boundary_ss}
	Note that from an application point of view restricting to non-negative real steady states is natural, as these are the ones with biological / chemical meaning. While in the CRN community it is standard to restrict to non-boundary steady states in the definition of multistationarity, in the case of fully open networks there is no need to argue as to whether we choose $\R_{\geq 0}^n$ or $\R_{> 0}^n$ as boundary steady states do not occur, which we see as follows. Assume $x^\star$ is a boundary steady state with $x_i^\star=0$. By the Hungarian Lemma \cite[Lemma 1]{Hungarian_Lemma_2020} all the terms with negative coefficients in the steady state polynomial corresponding to $\dot{x}_i$ contain the variable $x_i$ with positive exponent. At the same time, because of the inflow reaction of species $X_i$, there is a positive constant in this polynomial. That means after substituting $x_i=x^\star_i$ into this polynomial we end up with a polynomial with all coefficients positive and a constant term: this polynomial can not have any non-negative solutions on the remaining variables. This contradiction proves that the network could not have a boundary steady state to start with.
\end{remark}

The definition of multistationarity for the general case of CRNs can be found in \cite[Definition 1.1]{Sadeghimanesh-Feliu-2019}. When the network is fully open the two definitions are equivalent.\footnote{Simply because fully open networks do not have conservation laws.}

\begin{example}\label{ex:multistationairty_fully_open_example_plot}
	Consider the following network.
	\begin{align*}
		&A\ce{<=>[k_{2}][k_{1}]}0 \qquad B\ce{<=>[k_{4}][k_{3}]}0 \qquad C\ce{<=>[k_{6}][k_{5}]}0 \\
		&A+C\ce{->[k_{7}]}2A \qquad A+B\ce{->[k_{8}]}C
	\end{align*}
	Fix the reaction rate constants to the following values (which are provided by the software introduced later in Section~\ref{sec:multi_and_deficiency}):
	\begin{align*}
		&k_{1}=2.4956,\ k_{2}=1,\ k_{3}=11.1073,\ k_{4}=1,\\
		&k_{5}=1,\ k_{6}=1,\ k_{7}=3.5202,\ k_{8}=4.4904.
	\end{align*}
	Using the \texttt{RootFinding:-Isolate} command of \texttt{Maple} we solved the steady state equations to give us three solutions (given here to 3dp):
	\[(a,b,c)\approx (0.112,7.389,3.384),(0.626,2.915,2.870),(2.251,1.000,1.245).\]
	The network has more than one steady state for this choice of parameter values, which proves the multistationarity of the network.
\end{example}

\section{Approaches to detecting multistationarity}
\label{sec:multistationarity_methods}

\subsection{Examples}

We start with simple examples where the computations by hand are feasible.

\begin{example}\label{ex:detect_multi_manually}
	Consider the following fully open network with only one species.
	\begin{equation}\label{eq:toy_network_1}
		X\ce{<=>[k_2][k_1]}0\qquad 2X\ce{->[k_3]}3X.
	\end{equation}
	The ODE system associated to this network has only one equation:
	\[\dot{x}=k_1-k_2x+k_3x^2.\]
	To find whether this network is multistationary or not we should see if there exists any choice of $k\in\R_{>0}^3$ such that the steady state equation has more than one positive real solution. If we let $k=(6,5,1)$ then we find two positive real solutions $x^{(1)}=2$ and $x^{(2)}=3$. Therefore this network is multistationary.
	
	Now let us change the stoichiometric coefficient of $X$ in only one complex. Consider the following fully open network.
	\begin{equation}\label{eq:toy_network_2}
		X\ce{<=>[k_2][k_1]}0\qquad 2X\ce{->[k_3]}X.
	\end{equation}
	The ODE system associated to this network is
	\[\dot{x}=k_1-k_2x-k_3x^2.\]
	This time with a simple argument one can see that for any choice of the reaction rate constants, one of the two real solutions is always negative. So the network in this case has at most one steady state for all choices of parameters and therefore it is \emph{not} multistationary.
\end{example}

Even though in Example~\ref{ex:detect_multi_manually} the computation was straightforward, that is not the case in general. There do not exist an explicit formulae for solving polynomial equations in several variables and arbitrary degrees: to study these we usually require advanced techniques. Further, the approach needs to be algorithmic so that one can make a program to make the decision automatically. We continue this section by surveying four different approaches that are common in CRN theory for detecting multistationarity. We will also use these tools to generate datasets for ML training and validation in the subsequent sections.

\subsection{Algebraic tools}
\label{sec:multi_and_algebra}

An algebraic approach that in theory works for any given network consists of three steps. In the first step it receives the steady state polynomials of the network with the parameters treated as unknowns (not fixed to any specific values). It then uses elimination theory, for example the calculation of a Gr\"obner Basis (GB, as originally introduced by Buchberger \cite{Buchberger2006} with \cite{Cox-et-al-undergraduate-2015} a good introductory text), and returns a polynomial only involving the parameter unknowns. This polynomial is called the \emph{discriminant}. The set of vanishing points of the discriminant polynomial is a hypersurface in the parameter space. The parameter space for the fully open networks is $\mathbb{R}_{>0}^{|\mathcal{R}|}$ where $|.|$ is the cardinal and $\mathcal{R}$ is the set of reactions (both flow and non-flow reactions)\footnote{In the case of general networks the parameter space is $\mathbb{R}_{>0}^{|\mathcal{R}|+d}$ where $d$ is number of linearly independent conservation laws, i.e. $d=n-s$ where $n$ is the number of species and $s$ is the stoichiometric dimension (see Definition~\ref{def:stoichiometric_matrix_and_dimension}) of the network.}.

The second step uses  \emph{cylindrical algebraic decomposition} (CAD, originally introduced by Collins \cite{Collins1998b} with \cite{Jirstrand-1995} a good introductory text) to decompose the parameter space into a union of connected sets, each of which resides completely inside one connected component of the complement of the discriminant hypersurface. It is usual to perform only an \emph{open CAD}, which means to produce only full-dimensional open sets, accepting that the boundaries of the sets may contain unexpected behaviour.  
The CAD ensures that the number of (non-negative / positive) real solutions of the parametric polynomial system of equations is invariant in each of these connected components.  Thus it is enough to check the number of solutions of the system for a single sample point of each set. 

The third and final step is to repeatedly solve the system after substituting the values of the parameters to the sample points and count the number of solutions in each case.  The union of the regions where system has more than one non-negative / positive real solution is called the \emph{multistationarity region}. The network is multistationary if and only if the multistationarity region is non-empty. 

To read more about this approach and all its technical details see \cite{Lazard-Rouillier-2007,Moroz-PhD-thesis}; or for a simple explanation of how to use this algorithm in the context of multistationarity see \cite{Rost-Amir-2021,Sadeghimanesh-England}.  Additional papers that have developed this approach further include \cite{Rost-Amir-2023,Lichtblau-2021,Bradford-Davenport-England-Errami-et-al-2020}. This algorithm is also implemented in the \texttt{Maple} package \texttt{RootFinding:-Parametric} \cite{RootFinding-package}. We refer to the whole process --- computing the discriminant variety, then the CAD, then checking the number of solutions in each cell, then finally making a conclusion based on observing a sample point with more than one solution, or observing over all cells the system has at most one solution --- altogether as the \emph{algebraic approach}.

\subsection{Deficiency algorithms}
\label{sec:multi_and_deficiency}

The best known software in the CRN theory community is \texttt{CRNToolbox} \cite{CRNToolbox-2018}. This software uses deficiency theorems and algorithms to decide multistationarity of CRNs. 
We note this family of algorithm can not be applied to all networks, only those which satisfy particular properties.  
We refer the interested reader to \cite[Sections 2 and 3]{Joshi-2013} and \cite{Feinberg-1972,Feinberg-1987,Feinberg-1995,Ellison-1998,Haixia-2011} for the algorithms used in \texttt{CRNToolbox}.

\subsection{Inheritance of multistationarity via network extension operators}
\label{sec:multi_and_inheritance}

Another approach that attracted the attention of CRN theory researchers is to find under which conditions one can conclude properties of a CRN based on a smaller part of the network, i.e. when the larger network \emph{inherits} the property from the smaller. Since the computations for the smaller networks take less time, this can save a lot of computational effort. Some of these works involve removing reactions or removing species, or both, or partitioning these sets to smaller sets, see \cite{Banaji-2018,Banaji-Boros-Hofbauer-2023,Banaji-Pantea-2018,Feliu-Wiuf-2013,Sadeghimanesh-Feliu-2019,Banaji_2023}.

\begin{definition}\label{def:embedded_network}
	Let $\cN=(\cS,\cC,\cR)$ be a network, $S\subset\cS$ and $R\subset\cR$. Then perform the following steps in the order given.
	\begin{enumerate}
		\item Remove from $\cR$ any reactions not in $R$.
		\item Remove from $\cS$ any species not in $S$.
		\item Change the stoichiometric coefficients of the species that do not belong to $S$ in all complexes to 0 and update $\cC$ and $\cR$ accordingly. These are sets, so we do not keep duplicates.
		\item Remove any reaction $y \ce{->} y$ from $\cR$.
		\item Remove any complex that is not involved in any reaction from $\cC$.
		\item Remove any species that is not involved in any complex from $\cS$.
	\end{enumerate}
	The updated sets $\cS$, $\cC$ and $\cR$ define a new network that we denote by $\cN|_{S,R}$ and call the \emph{embedded subnetwork} of $\cN$ induced by $S$ and $R$, (alternatively it can be called the embedded subnetwork of $\cN$ after removal of $\cS-S$ and $\cR-R$ from $\cN$). The network $\cN$ is called an \emph{extension} of $\cN|_{S,R}$.
\end{definition}

\begin{example}\label{ex:embedded_subnetworks}
	Consider the network $\cN_1$ with the CR-graph as follows. It has 7 species, 9 complexes and 8 reactions.
	
	\[\begin{array}{c}
		X_1+X_2\ce{->[k_1]}X_1+X_3\ce{->[k_2]}X_1+X_4\ce{->[k_3]}X_1+X_5\ce{->[k_4]}X_1+X_6\\
		X_3+X_4\ce{->[k_5]}X_1\ce{<=>[k_6][k_7]}X_7\ce{->[k_8]}X_2+X_5
	\end{array}\]
	
	We want to generate an embedded subnetwork, $\cN_2$, from $\cN_1$ by removing two species $X_3$ and $X_4$ and two reactions: the 4th and the 7th. Applying the first step the CR-graph becomes the following.
	
	\[\begin{array}{c}
		X_1+X_2\ce{->}X_1+X_3\ce{->}X_1+X_4\ce{->}X_1+X_5\qquad X_1+X_6\\
		X_3+X_4\ce{->}X_1\ce{->}X_7\ce{->}X_2+X_5
	\end{array}\]
	
	\noindent Applying Steps 2 and 3 then gives us the following CR-graph.
	
	\[\begin{array}{c}
		X_1+X_2\ce{->}X_1\ce{->}X_1\ce{->}X_1+X_5\qquad X_1+X_6\\
		0\ce{->}X_1\ce{->}X_7\ce{->}X_2+X_5
	\end{array}\]
	
	\noindent Applying Step 4 updates the CR-graph as follows.
	
	\[\begin{array}{c}
		X_1+X_2\ce{->}X_1\ce{->}X_1+X_5\qquad X_1+X_6\\
		0\ce{->}X_1\ce{->}X_7\ce{->}X_2+X_5
	\end{array}\]
	
	\noindent Finally after Steps 5 and 6 we have $\cN_2$ below.
	
	\[\begin{array}{c}
		X_1+X_2\ce{->}X_1\ce{->}X_1+X_5\\
		0\ce{->}X_1\ce{->}X_7\ce{->}X_2+X_5
	\end{array}\]
	
	Note that even though the species $X_6$ was not requested to be removed, as the result of removing the fourth reaction in the original network, this species automatically is removed. The new network has 4 species, 6 complexes and 5 reactions. Also note that the zero complex and the complex $X_1$ (containing only $X_1$ with stoichiometric coefficient one) did not exist in the original network, but by removal of species 3 and 4 they are generated. The CR-graph of $\cN_2$ has only one linkage class while $\cN_1$ had two.
	
	\[\xymatrix @C=1pc @R=1pc{
		& 0 \ar@{>}[d] & \\
		X_1+X_2 \ar@{>}[r] & X_1 \ar@{>}[r] \ar@{>}[d] & X_1+X_5\\
		& X_7 \ar@{>}[r] & X_2+X_5
	}\]\smallskip
	
\end{example}

\begin{thm}\label{thm:inheritance}
	\emph{(\cite[Corollary 4.6]{Joshi-Shiu-2013}).} Let $\cN$ be a fully open network and $\cN'$ a fully open embedded subnetwork of $\cN$. If $\cN'$ is multistationary, then $\cN$ is multistationary as well.
\end{thm}

This theorem tells us that by adding reactions or species to a fully open multistationary network we again obtain a multistationary network. However, we can not conclude anything about multistationarity of a subnetwork of a multistationary network. Similarly we can conclude that removing reactions and species from a fully open non-multistationary network, gives us a non-multistationary network, but we can not conclude anything about multistationarity of an extension of a fully open non-multistationary network.

\begin{example}\label{ex:inheritance_of_multistationarity}
	Since Network \eqref{eq:toy_network_1} is multistationary, the following network is also multistationary by Theorem~\ref{thm:inheritance}.
	\[X_1\ce{<=>}0\ce{<=>}X_2\quad 2X_1+X_2\ce{->}3X_1\]
\end{example}

A network is called an \emph{atom of multistationarity} if it is multistationary and all of its embedded subnetworks are non-multistationary.

\subsection{Positive feedback loops}
\label{sec:multi_and_feedback_loops}

The last approach that we will discuss is a method to preclude multistationarity. This needs a different graph representation of a CRN called the \emph{directed species-reaction graph} (DSR-graph).

\begin{definition}\label{def:influx_matrix}
	Let $\cN$ be a network with $n$ species, $X_i$, and $r$ reactions with the occurrence rates $\rho_i$. Note that each $\rho_i$ can be a function of concentrations of the species, as in the mass action kinetics. The \emph{influx matrix}, $Z=\begin{bmatrix}
		z_{i,j}
	\end{bmatrix}_{n\times r}$ is defined as a symbolic matrix with entries as
	\begin{equation}\label{eq:influx_matrix_entries}
		z_{i,j}=\sgn \left( \frac{d\rho_j}{dx_i} \right)\gamma_{i,j}.
	\end{equation}
\end{definition}

In the case of mass action kinetics (which is what we restrict ourselves to it in this manuscript) the influx matrix has a simple form stated in the following lemma whose proof is straightforward.

\begin{lem}\label{lem:influx_matrix_for_mass_action}
	Let $\cN$ be a network with $n$ species, $X_i$, and $r$ reactions, which is equipped with mass action kinetics. Equation~\eqref{eq:influx_matrix_entries} simplifies to $z_{i,j}=\gamma_{i,j}$ if $X_i$ is at the reactant side of the $j$th reaction, otherwise $z_{i,j}=0$.
\end{lem}

\begin{definition}\label{def:stoichiometric_matrix_and_dimension}
	Let $\cN$ be a network. The \emph{stoichiometric matrix}, $N$, is  the matrix whose columns are the vectors $y'-y$ for the reactions $y \ce{->} y'$ of $\cN$. The rank of $N$ is called the \emph{stoichiometric dimension} of $\cN$ and is denoted by $s$.
\end{definition}

\begin{definition}\label{def:DSR-graph}
	Let $\cN$ be a network with $\cS=\lbrace X_1,\dots,X_n\rbrace$, $\cR=\lbrace\fr_1,\dots,\fr_r\rbrace$, the stoichiometric matrix $N=\begin{bmatrix}
		a_{i,j}
	\end{bmatrix}_{n\times r}$, and the influx matrix $Z=\begin{bmatrix}
		z_{i,j}
	\end{bmatrix}_{n\times r}$. The \emph{DSR-graph} is a labelled digraph with $\cS\cup\cR$ as its set of nodes. For every non-zero entry of $N$, there is an edge from $\fr_j$ to $X_i$ with $a_{i,j}$ as its label, and for every non-zero entry of $Z$, there is an edge from $X_i$ to $\fr_j$ labelled $z_{i,j}$.
	
	Additionally the \emph{signed DSR-graph} is the same as the DSR-graph with the difference that the $\gamma_{i,j}$s are removed from the labels coming from the influx matrix, and only the signs behind them remain.  Similarly the labels coming from the stoichiometric matrix are replaced with their signs, that is $\sgn(a_{i,j})$. In this representation of the signed DSR-graphs we write $+$ and $-$ instead of $1$ and $-1$ on the edges.
\end{definition}

\begin{example}\label{ex:DSR_graph_example}
	Recall Network \eqref{eq:toy_network_2} from Example~\ref{ex:detect_multi_manually}. The DSR-graph and signed DSR-graph are shown in Figure~\ref{fig:DSR_graph_example}.
	
	\begin{figure}[ht]
		\begin{tabular}{cc}
			\resizebox{0.5\textwidth}{!}{
				\begin{minipage}{0.5\linewidth}
					\[\xymatrix @C=3pc @R=2.5pc{
						& \fr_1 \ar@/_1pc/[dl]_1 \\
						X \ar@/^0.5pc/[r]^{\gamma_{1,2}} \ar@/^0.2pc/[dr]|-{\gamma_{1,3}} & \fr_2 \ar@/^0.5pc/[l]_{-1}\\
						& \fr_3 \ar@/^1pc/[ul]^{-1}
					}\]
				\end{minipage}
			}
			&
			\resizebox{0.5\textwidth}{!}{
				\begin{minipage}{0.5\linewidth}
					\[\xymatrix @C=3pc @R=2.5pc{
						& \fr_1 \ar@/_1pc/[dl]_{+} \\
						X \ar@/^0.5pc/[r]^{+} \ar@/^0.2pc/[dr]|-{+} & \fr_2 \ar@/^0.5pc/[l]_{-}\\
						& \fr_3 \ar@/^1pc/[ul]^{-}
					}\]
				\end{minipage}
			}
			\vspace{0.2cm}\\
			(a) & (b)
		\end{tabular}
		\caption{(a) DSR-graph and (b) signed DSR-graph for Network \eqref{eq:toy_network_2} of Example~\ref{ex:detect_multi_manually}.}
		\label{fig:DSR_graph_example}
	\end{figure}
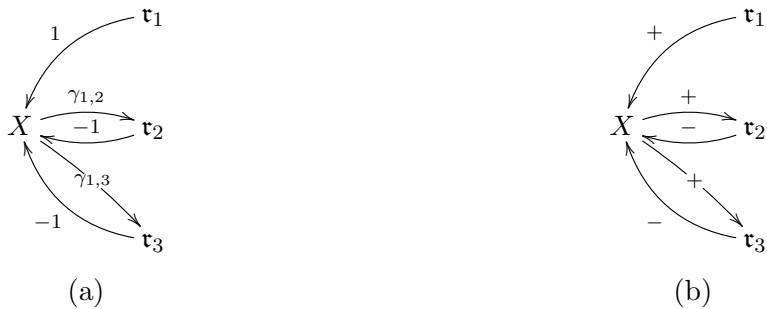
	
\end{example}

Recall that a \emph{cycle} in a graph is a closed path that does not cross itself at any vertices.

\begin{definition}\label{def:feedback_loops}
	Consider the signed DSR-graph of a network. The label or sign of a cycle is the product of labels from its edges. A cycle with positive label is called a \emph{positive feedback loop}. Similarly a \emph{negative feedback loop} is a cycle with a negative label. The length of a feedback loop is the number of species vertices used in this cycle. A feedback loop is called \emph{even} or \emph{odd} if its length is an even or odd integer respectively. Two feedback loops are called \emph{disjoint} if they have no common nodes.
\end{definition}

\begin{definition}\label{def:nucleus}
	A \emph{nucleus} of length $m$ is a set of $m$ disjoint feedback loops. Let $D=\lbrace L_1,\dots,L_m\rbrace$ be a nucleus and assume $m'$ is the number of even feedback loops in $D$. The label or sign of $D$ is defined as
	\[\sgn(D)=(-1)^{m'}\prod_{i=1}^m\sgn(L_i).\]
\end{definition}

\begin{thm}\label{thm:positive_feedback_loops}
	\emph{(\cite{Feliu-Wiuf-2015}).} Let $\cN$ be network with stoichiometric dimension $s$. If the sign of every nucleus of length $s$ is equal to $(-1)^s$, then $\cN$ is not multistationary. 	
\end{thm}

The following lemma is implicitly proved in \cite{Feliu-Wiuf-2015}.

\begin{lem}\label{lem:sign_of_negative_feedback_loops}
	\emph{(\cite{Feliu-Wiuf-2015}).} If all feedback loops in a nucleus $D$ are negative, and the length of $D$ is $m$, then $\sgn(D)=(-1)^m$.
\end{lem}

Theorem~\ref{thm:positive_feedback_loops} and Lemma~\ref{lem:sign_of_negative_feedback_loops} together give the following corollary.

\begin{corollary}\label{cor:no_positive_feedback_loop}
	\emph{(\cite{Feliu-Wiuf-2015}).} If a network does not contain any positive feedback loop, then it is not multistationary.
\end{corollary}

Note that there are only two feedback loops in Example~\ref{ex:DSR_graph_example}, both negative. Therefore, by Corollary~\ref{cor:no_positive_feedback_loop}, Network \eqref{eq:toy_network_2} is not multistationary. When there exists a nucleus whose sign is different from $(-1)^s$, Theorem~\ref{thm:positive_feedback_loops} is of no use for deciding multistationarity of the network. For example it does not make any conclusion about the multistationarity of Network \eqref{eq:toy_network_1}. The signed DSR-graph is the same as in Figure~\ref{fig:DSR_graph_example}(b), but the label of the edge from $\fr_3$ to $X$ is now positive. It has stoichiometric dimension 1 and two feedback loops, one negative of length one and one positive of length one. The positive feedback loop alone forms a nucleus of length 1 with no even feedback loops so its sign is 
\[(-1)^0(+1)^1=1,\]
while $(-1)^s=(-1)^1=-1$. Thus there is a nucleus with a differing sign and Theorem~\ref{thm:positive_feedback_loops} can not rule on its non-multistationarity, at the same time we can not say it will be multistationary just because the theorem did not label it as non-multistationary.

We now go further than \cite{Feliu-Wiuf-2015} and prove the following new proposition which is simplified version of Theorem~\ref{thm:positive_feedback_loops} for fully open networks equipped with mass action kinetics.

\begin{proposition}\label{prop:positive_feedback_loops_for_fully_open}
	Theorem~\ref{thm:positive_feedback_loops} is inconclusive for a fully open network equipped with mass action kinetics if and only if there exists a positive feedback loop. I.e. when there exists a positive feedback loop, there is no need to check label of nuclei.
\end{proposition}

Before presenting the proof of Proposition~\ref{prop:positive_feedback_loops_for_fully_open}, we include a straightforward corollary of the result.

\begin{corollary}\label{cor:smallest_positive_feedback_loops}
	If a fully open network equipped with mass action kinetics contains a reaction with a species on both sides of it, but with higher stoichiometric coefficient on the product side, then Theorem~\ref{thm:positive_feedback_loops} is inconclusive.
\end{corollary}

\begin{proof}
	Let the species in the corollary be $X$ with the stoichiometric coefficients $a$ and $b$ in the reactant and the product sides of $\fr$ respectively. By the corollary's assumption $1\leq a<b$. Therefore signs of the corresponding entry of the influx matrix and the stoichiometric matrix are both positive and
	$\xymatrix @C=3pc @R=2.5pc{
		X \ar@/^0.2pc/[r]^{+} & \fr \ar@/^0.2pc/[l]^{+}
	}$
	is a positive feedback loop. By Proposition~\ref{prop:positive_feedback_loops_for_fully_open}, Theorem~\ref{thm:positive_feedback_loops} does not decide multistationarity status of the network.	
\end{proof}

In particular, if a fully open network contains a reversible reaction with a species appearing on both sides with different stoichiometric coefficients, one of the two sides gives us a positive feedback loop. Therefore Theorem~\ref{thm:positive_feedback_loops} does not decide the multistationarity of such a network.

\begin{proof}[Proof of Proposition \ref{prop:positive_feedback_loops_for_fully_open}]
	Since the network is fully open the stoichiometric dimension is the number of species, $n$. For every species the outflow reaction gives a negative feedback loop of length one.
	
	If no positive feedback loop exists then by Lemma~\ref{lem:sign_of_negative_feedback_loops}, the condition of Theorem~\ref{thm:positive_feedback_loops} holds. Conversely assume there exists a positive feedback loop. Denote the length of this positive feedback loop with $m$. If we select this positive feedback loop with $m-n$ negative feedback loops of length one for outflow reactions of the species not involved in the previous positive feedback loop, then the sum of the lengths is $n$ and the selected feedback loops all are disjoint. Therefore we have a nucleus whose sign is $\big((-1)^1\big)^{n-m}\big((-1)^0(+1)^m\big)=(-1)^{n-m}$ if $m$ is odd, and is $\big((-1)^1\big)^{n-m}\big((-1)^1(+1)^m\big)=(-1)^{n-m+1}$ if $m$ is even. If $m$ is odd, $n-m$ is odd or even if $n$ is even or odd respectively. That means $(-1)^{n-m}$ is always of opposite sign to $(-1)^n$. If $m$ is even, then $m-1$ is odd and with a same reasoning as in the previous sentence, the sign of $(-1)^{n-m+1}=(-1)^{n-(m-1)}$ is always of opposite sign to $(-1)^n$. Thus the condition of Theorem~\ref{thm:positive_feedback_loops} does not hold.
\end{proof}

We define a new operation on fully open networks that changes some stoichiometric coefficients.

\begin{definition}
	\label{def:negitive_friendly_operation}
	Let $\cN$ be a CRN (not necessarily fully open) with $\cS=\lbrace X_1,\dots,X_n\rbrace$ and $\cR=\lbrace \fr_1,\dots,\fr_r\rbrace$ where $\fr_i=\sum_{j=1}^na_j^{(i)}X_j\ce{->}\sum_{j=1}^nb_j^{(i)}X_j$. Define a new set of reactions $\cR'=\lbrace \fr_1',\dots,\fr_r'\rbrace$ with $\fr_i'=\sum_{j=1}^nc_j^{(i)}X_j\ce{->}\sum_{j=1}^nd_j^{(i)}X_j$ satisfying the following two conditions.
	\begin{itemize}
		\item For every $i$ and $j$, $a_j^{(i)}=0$ if and only if $c_j^{(i)}=0$, similarly $b_j^{(i)}=0$ if and only if $d_j^{(i)}=0$.
		\item For every $i$ and $j$, if $a_j^{(i)}\prec b_j^{(i)}$ where $\prec$ is one of the three relations $<$, $=$, $>$, then $c_j^{(i)}\prec d_j^{(i)}$.
	\end{itemize}
	Let $\cN'$ be the CRN defined by $\cR'$ as its set of reactions. We call this operation that converts $\cN$ to $\cN'$, a \emph{negative friendly operation} for the reason that will be revealed at Proposition~\ref{prop:changing_sc_without_loosing_NPFL} below.
\end{definition}

\begin{example}
	\label{ex:negative_friendly_operation}
	Consider a fully open network with four species $X_i$, $i=1,2,3,4$ and the following three non-flow reactions.
	\begin{equation}
		\label{eq:negative_friendly_operation_network_1}
		2X_1\ce{->[k_1]}X_1+X_2,\quad X_2+X_3\ce{->[k_2]}X_4,\quad 2X_4\ce{->[k_3]}0.
	\end{equation}
	The negative friendly operation tells us that in the first reaction we only can have $X_1$ in the reactant side and it must appears in the product side too, but their stoichiometric coefficients can be arbitrary positive integers as far as the coefficient on the reactant side is greater than in the product side, for example we can choose to increase the coefficient from 2 to 3 in the reactant side and keep the coefficient on the product side equal to 1. For the product complex we should also have $X_2$ with positive coefficient, but since there is no $X_2$ in the reactant side, we have no other constraint and it can be any positive integer. Let us increase this coefficient from 1 to 2. So the first reaction is converted to $3X_1\ce{->[k_1]}X_1+2X_2$. Similarly it is permitted under the negative friendly operation's rules to change the second reaction to $X_2+X_3\ce{->[k_2]}2X_4$ and the third reaction to $X_4\ce{->[k_3]}0$. However the last reaction will be a duplicate of the outflow reaction of $X_4$ and thus we disregard it. The result of the operation is a new fully open network with still four species but only two non-flow reactions shown below.
	\begin{equation}
		\label{eq:negative_friendly_operation_network_2}
		3X_1\ce{->[k_1]}X_1+2X_2,\quad X_2+X_3\ce{->[k_2]}2X_4.
	\end{equation}
\end{example}

\begin{proposition}
	\label{prop:changing_sc_without_loosing_NPFL}
	Let $\cN$ be a CRN (not necessarily fully open) equipped with mass action kinetics that does not have any positive feedback loop, and assume that $\cN'$ is a transformation of it via a negative friendly operation. Then $\cN'$ also does not have any positive feedback loop.
\end{proposition}

\begin{proof}
	Since $\cN'$ is a transformation of $\cN$ via a negative friendly operation, there exists a map $\phi\colon\cR\rightarrow\cR'$ where $\phi(\fr)$ is made by changing stoichiometric coefficients  of some (or none) of the species in $\fr$ without violating the two rules in Definition~\ref{def:negitive_friendly_operation}.
	Name the signed DSR-graphs of $\cN$ and $\cN'$ as $G$ and $G'$ respectively. If the change in the stoichiometric coefficients does not cause two reactions to become identical, then $G'=G$ and the proposition trivially holds. If two reactions become redundant, say $\phi(\fr_i)=\phi(\fr_j)$, for $i\neq j$, then $G'$ is the induced subgraph of $G$ obtained by removal of the vertex $v_{\fr_j}$. Thus the set of positive feedback loops of $\cN'$ is a subset of positive feedback loops of $\cN$. This completes the proof.
\end{proof}

\begin{corollary}
	\label{cor:infinite_non_multi}
	Let $\cN$ be a fully open network equipped with mass action kinetics and without a positive feedback loop, then any transformation of $\cN$ by a negative friendly operation is a non-multistationary network.
\end{corollary}

Note that this new operation keeps the set of species the same, but changes the sets of complexes and reactions. Furthermore, while each atom of multistationarity only provides finitely many choices of non-multistationary networks, a network without positive feedback loop provides us infinitely many choices of non-multistationary networks.

\subsection{Limitation of current approaches}
\label{sec:example_for_limitation_of_4_multi_methods}

To show the limitations of the current methods we study the following example.

\begin{example}\label{ex:example_for_limitation_of_current_methods}
	Consider a fully open network with 6 species and 6 non-flow reactions as follows.
	\begin{equation}\label{eq:network_for_limitation_of_current_methods}
		\begin{array}{ccc}
			X_1 + X_2 & \ce{->} & X_3\\
			X_2 + X_3 & \ce{<=>} & X_1 + X_4\\
			X_3 + X_5 & \ce{->} & 2 X_6\\
			X_4 & \ce{->} & X_5\\
			X_1 + 2 X_5 & \ce{->} & X_1
		\end{array}
	\end{equation}
	Giving this network to \texttt{CRNToolbox} as an input, the initial report suggests to run the higher deficiency algorithm. Note that the CR-graph contains the flow reactions as well, this graph has three linkage classes and the deficiency of the network is 4. After requesting the higher deficiency algorithm from \texttt{CRNToolbox}, which took more than a minute to finish the computation, the report states ``\textit{taken with mass action kinetics, the network might have the capacity for multiple steady states}''. I.e. The higher deficiency algorithm failed to determine the multistationarity of this network.  
	
	The report continues with a suggestion to run other algorithms implemented in  \texttt{CRNToolbox}: ``\textit{to determine whether the network might admit two distinct stoichiometrically compatible positive equilibria, you should also consult the mass action injectivity or concordance report if you haven't already done so. (To test for rate constants that give a degenerate steady state, consult the Zero Eigenvalue Report)}''. The mass action injectivity report is produced instantly but failed to make a decision on multistationarity. The report says ``\textit{taken with mass action kinetics, the network is not injective}''. 
	
	We then turn to the other suggestions in the higher deficiency report. When asking \texttt{CRNToolbox} for the zero eigenvalue report, the CPU use of the computer increases considerably and even after 10 minutes waiting the software does not give any output. The last suggestion of \texttt{CRNToolbox}, the concordance report, states that there is a kinetics such that the network can be multistationary, however the given kinetics is not mass-action. In summary, \texttt{CRNToolbox} and its deficiency results (the higher deficiency algorithm in this case) failed to determine multistationarity of Network~\eqref{eq:network_for_limitation_of_current_methods}.
	
	We attempt to use the algebraic approach via \texttt{Maple} package, \texttt{RootFinding:-Parametric}.  The amount of memory used quickly passes 25 GB within 7 minutes with the computation still unfinished.\footnote{Processor: Intel(R) Core(TM) i7-10850H CPU \@2.70GHz 2.71 GHz. Installed memory (RAM): 64.0 GB (63.6 GB usable). System type: 64-bit Operating System, x64-based processor.}
	
	Regarding the inheritance method: the network does not contain any atom of multistationarity that we know of and so we cannot apply any of the known inheritance results (for fully open networks with one reversible or irreversible non-flow reaction, or fully open networks with two reversible and irreversible non-flow reactions where the sum of the stoichiometric coefficients of each complex is at most two, as introduced in \cite{Joshi-2013,Joshi-Shiu-2013}).
	Hence this method also fails to detect multistationarity of this network.
	
	Finally the last approach, the use of positive feedback loops for excluding multistationarity also fails. Note that by considering the reversible reaction $X_2 + X_3 \ce{<=>} X_1 + X_4$ and naming the forward reaction as $\fr_1$ and the backward reaction as $\fr_2$, we get the following positive feedback loop.
	\[\xymatrix @C=3pc @R=0.5pc{
		& \fr_1 \ar@/^0.5pc/[dr]^+ & \\
		X_2 \ar@/^0.5pc/[ru]^+ &  & X_1 \ar@/^0.5pc/[ld]^+\\
		& \fr_2 \ar@/^0.5pc/[ul]^+ &
	}\]
	Thus by Proposition~\ref{prop:positive_feedback_loops_for_fully_open}, Theorem~\ref{thm:positive_feedback_loops} is inconclusive for Network~\eqref{eq:network_for_limitation_of_current_methods}.
\end{example}

This example motivates us to look for an alternative approach with ML.

\section{Limitations of simple ML approaches}
\label{sec:ML}

An observation, true for any network, is that when we fix the edge set of a CR-graph of a network and only let the stoichiometric coefficients vary in the complexes appearing in the vertex set, then multistationarity of the network is dependent on these stoichiometric coefficients. Therefore a natural question to ask is whether we can train a Machine Learning (ML) model to identify whether a network is multistationary or non-multistationary based on those coefficients.

In some cases there are existing results  making explicit the relationship.  For example, for the smallest fully open networks, namely fully open networks with only one irreversible or reversible non-flow reaction, Joshi has completely classified multistationarity of these networks in \cite{Joshi-2013} using deficiency results (Section~\ref{sec:multi_and_deficiency}).

\begin{thm}\label{thm:one_non_flow_Joshi_thm}
	\emph{(\cite[Theorem 4.1]{Joshi-2013}).} Let $\cN$ be a fully open network with $\cS=\lbrace X_1,\dots,X_n\rbrace$ and only one non-flow reaction as follows.
	\[a_1X_1+\dots+a_nX_n\ce{->}b_1X_1+\dots+b_nX_n.\]
	The network $\cN$ is multistationary if and only if
	\[\sum_{\substack{i=1\\a_i<b_i}}^na_i>1.\]
	If the non-flow reaction is reversible then the above condition is replaced by
	\[\sum_{\substack{i=1\\a_i<b_i}}^na_i>1\text{ or }\sum_{\substack{i=1\\b_i<a_i}}^nb_i>1.\]
\end{thm}

Let us consider whether ML could identify such results. Consider fully open networks as in Theorem~\ref{thm:one_non_flow_Joshi_thm}, with only one species and one irreversible non-flow reaction. So the only varying features of the networks are the stoichiometric coefficients of this species on the two sides of the non-flow reaction, $aX \ce{->} bX$.  Then the input space for the ML problem can be the two-dimensional space, $\lbrace (a,b)\in\N^2\mid a\neq b\rbrace$. By Theorem~\ref{thm:one_non_flow_Joshi_thm} the set of multistationary networks will correspond to the set of points in the input space with $a>\frac{3}{2}$ and $b>a$. 
This suggests that we may use a relatively simple ML technique such as Support Vector Machines (SVMs) \cite{Cristianini-Shawe_Taylor-2000} to fit a dividing hypersurface between the classes.  

However, one should note that models such as SVM and Random Forest (RF) \cite[Chapter 5]{Zhang-Ma-2012} require fixed length input vectors, whereas the number of stoichiometric coefficients of networks vary based on the number of reactions and species. So although we can get nice results for small sub-cases this way, it does not generalise to study broader families of CRNs. In the next section we develop new methods to handle variable length input data for CRNs.

Interested readers can find examples of training SVM and RF for fixed size small fully open networks within the code release supporting this paper: see the Data Access Statement at the end of the paper for the link.

\section{ML for CRNs via a new graph representation}
\label{sec:GAT}

We opt for a graph embedding of our chemical reaction networks as input to ML, since this seems best placed to capture the detail of the interactions involved. There are a variety of ML approaches that accept graphs as input, usually called graph learning algorithms. Probably the best known is the \emph{Graph Neural Network} (GNN) which consists of an iterative process to propagate node states until equilibrium, after which a neural network is trained to produce an output for each node.  An overview of the family of GNN methods may be found in \cite{Liu-Zhou-2020}. 

However, before applying any such tool we need to first decide precisely how to present a CRN as a graph input. For this purpose we introduce a new graph representation for a fully open network, designed to allow for decisions on multistationarity based on the stoichiometric coefficients and direction of the reactions.

\subsection{New graph representation of CRNs}

Since the species flow reactions are included in all fully open CRNs, they do not play a role in determining the multistationarity ones from the non-multistationarity.  Thus we will only encode the rest of the reactions for our ML model.

\begin{definition}\label{def:SRSC_graph}
	Let $\cN$ be a fully open network with $\cS=\lbrace X_1,\dots,X_n\rbrace$ and assume there are $r$ reactions other than the species flows where the $i$th such reaction, denoted by $\fr_i$, is written as
	\[\sum_{j=1}^na_j^{(i)}X_j\ce{->}\sum_{j=1}^nb_j^{(i)}X_j.\]
	We construct a new weighted labelled digraph representation with the vertex set $V$ and edge set $E$ following the steps below.
	\begin{enumerate}
		\item For any reaction $\fr_i$ construct a vertex $v_{\fr_i}\in V$ to represent this reaction.
		\item For every species in the reactant side of $\fr_i$, say $X_j$, construct a vertex $\xheightsub{v}{X_j\to\fr_i}$. That means if $a_j^{(i)}\neq 0$ then $\xheightsub{v}{X_j\to\fr_i}\in V$. Construct also a directed edge $\xheightsub{e}{X_j\to\fr_i}=(\xheightsub{v}{X_j\to\fr_i},v_{\fr_i})\in E$.
		\item Similarly for every species in the product side of $\fr_i$, say $X_j$, we associate a vertex $\xheightsub{v}{\fr_i\to X_j}$. That means if $b_j^{(i)}\neq 0$ then $\xheightsub{v}{\fr_i\to X_j}\in V$.  Construct also directed edge $\xheightsub{e}{\fr_i\to X_j}=(\xheightsub{v}{\fr_i},\xheightsub{v}{\fr_i\to X_j})\in E$.
		\item Connect all vertices of the form $\xheightsub{v}{X_j\to\fr_i}$ or $\xheightsub{v}{\fr_i\to X_j}$ via reversible edges with the label $X_j$.
		\item Assign every vertex of the form $v_{\fr_i}$ the weight 0.
		\item Assign all vertices of the form $\xheightsub{v}{X_j\to\fr_i}$ the weight $a_j^{(i)}$, and all vertices of the form $\xheightsub{v}{\fr_i\to X_j}$ the weight $b_j^{(i)}$.
	\end{enumerate}
	Since a species with 0 as its stoichiometric coefficient in a complex does not contribute to a vertex in this graph, the only vertices with 0 as their weight are the reaction vertices. All other vertices contain one of the non-zero stoichiometric coefficients that appear in the CR-graph of the network. The labels on the reversible reactions make it clear which of the stoichiometric coefficients in this new graph are coming from the same species and makes it possible to recreate the complexes. 
	
	Note that a permutation on the species names or a renaming of them does not change the multistationarity of the network. The irreversible unlabelled edges show which side of the reaction the species  appeared, with the given stoichiometric coefficient. Therefore the CR-graph can uniquely, up to a renaming of the species, be rewritten from this new graph. We call this new graph representation of a fully open network the \emph{\textbf{s}pecies-\textbf{r}eaction graph with emphasis on \textbf{s}toichiometric \textbf{c}oefficients} abbreviated as the \emph{SRSC-graph}.
\end{definition}

\begin{example}\label{ex:SRSC_graph}
	Let $\cN_1$ be a fully open network with two species $A$ and $B$ and one non-flow reversible reaction, \begin{equation}\label{eq:network_for_SRSC_example_1}
		A+2B\ce{<=>}3A+B.
	\end{equation}
	Let $\cN_2$ be a fully open network with three species $A$, $B$ and $C$ and three non-flow irreversible reactions:
	\begin{equation}\label{eq:network_for_SRSC_example_2}
		A+B\ce{->}2A,\;A+C\ce{->}B,\;2C\ce{->}B+C.
	\end{equation}
	The SRSC-graph of these two networks are presented in Figure~\ref{fig:SRSC_graph}.
	
	\begin{figure}[ht]
		\centering
		\begin{tabular}{cc}
			\includegraphics[width=6cm]{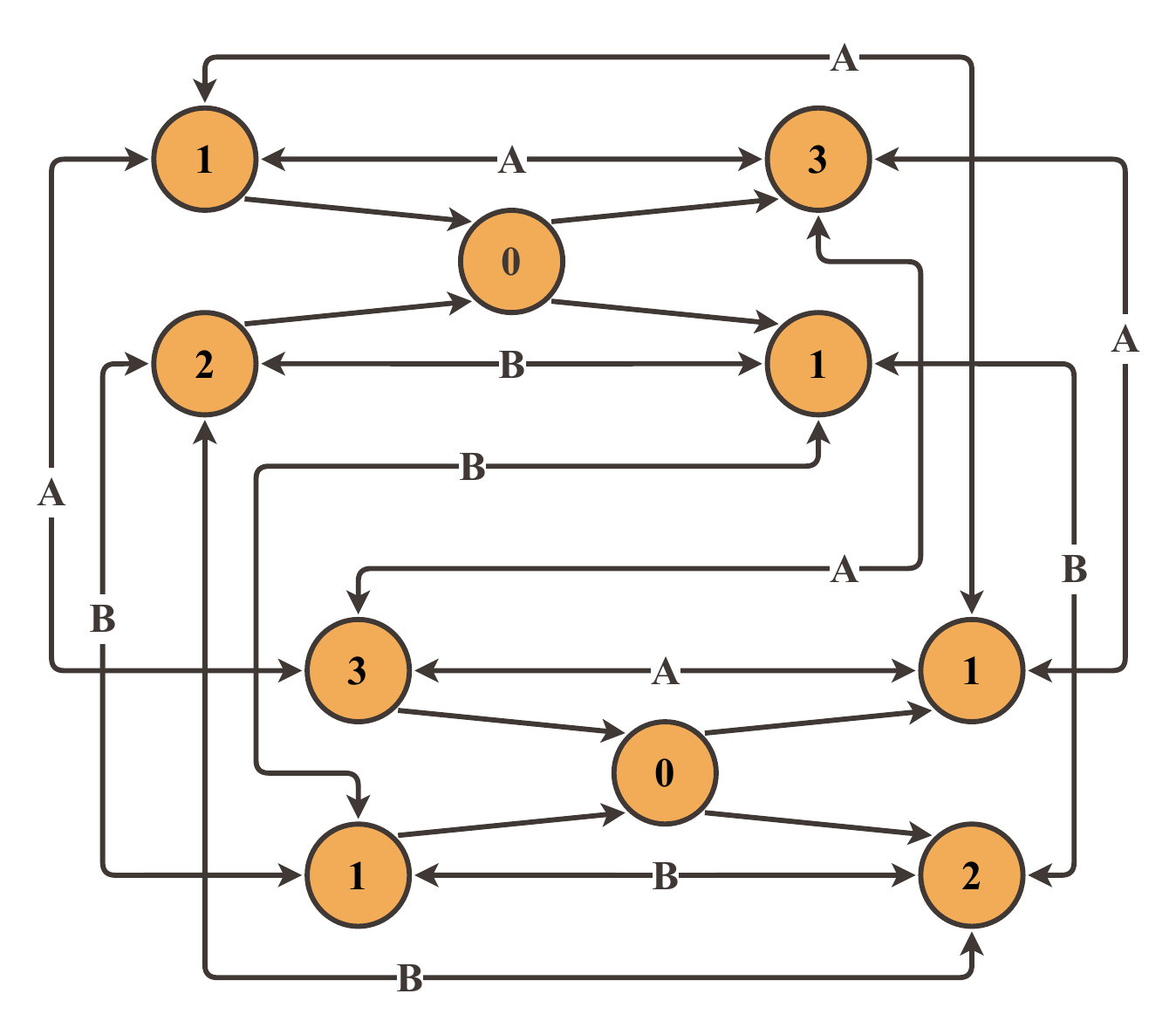} &
			\includegraphics[width=6cm]{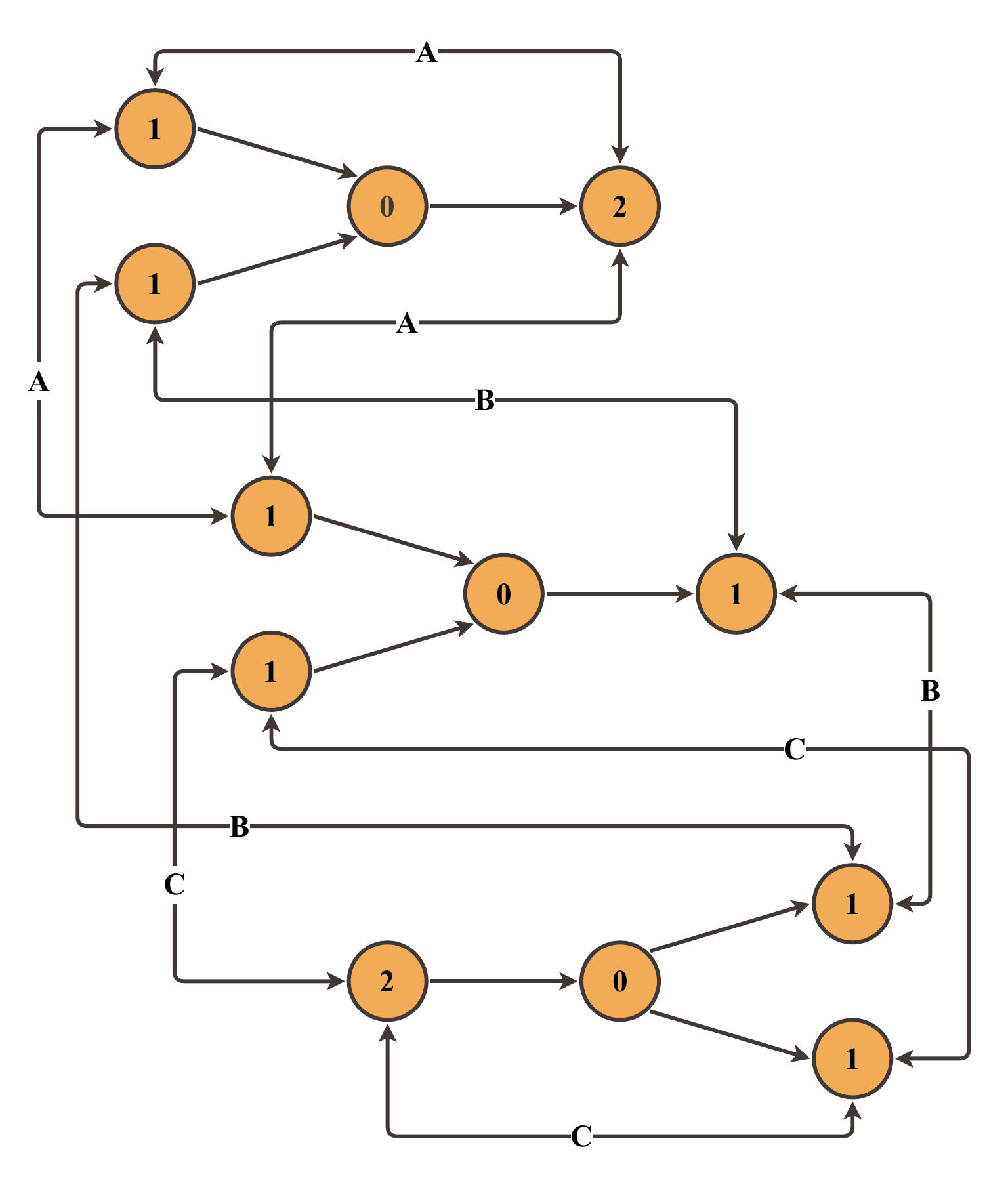}\\
			(a) & (b)\\
		\end{tabular}
		\caption{The SRSC-graph representation of  (a) Network \eqref{eq:network_for_SRSC_example_1} and (b) Network \eqref{eq:network_for_SRSC_example_2}.}
		\label{fig:SRSC_graph}
	\end{figure}
	
\end{example}

To see the reasoning behind Definition~\ref{def:SRSC_graph} let us briefly go through what happens in a graph learning algorithm. A graph learning algorithm works as follows. The input is a set of nodes and edges. Each node has a scalar or a vector of numbers as its feature; in the case of vectors, all must be of the same length. The algorithm updates the value of each node using an aggregation function, such as average, sum, or a learnable function, that combines its own previous value with the values from its neighbours. In directed graphs, only neighbours with incoming edges to the node are considered. This process is repeated for several steps (layers), and the algorithm can learn to assign different weights to different neighbours during aggregation.

For graph-level classification (not node-level), the final node representations are combined using a pooling operation to produce a single graph representation: examples include mean, sum, or max pooling. This graph-level vector is then passed to a classifier (for example a neural network) that works with fixed-length input to make the prediction.

The decisive data in our problem are the stoichiometric coefficients, and therefore we created nodes containing these scalar values. We decided to allow aggregation in the direction of reactants to products, and also freely among different instances of a species in different complexes. Note that without the later permission, the graph may end up with low connectivity, leading the graph learning algorithm to perform poorly. We also put 0 as the value on the reaction vertices since they are not encoding stoichiometric coefficients and they are only playing the role of a bridge between reactants and products and thus they should not have a numerical contribution.

\subsection{Dataset creation}
\label{subsec:dataset}

Our new graph representation allows us to apply graph learning algorithms and, in contrast to SVM and RF, we do not need to train a new model for each choice of number of species and number of non-flow reactions. However, graph learning networks need large datasets for training and testing. In addition to that, our goal is not to only study one non-flow reaction since Theorem~\ref{thm:one_non_flow_Joshi_thm} has already provided an exact and easy answer to that case. We want to have a tool to predict multistationarity of a general fully open network and so we need to create a dataset that has representation from a reasonably diverse set of shapes of fully open networks. To create this labelled dataset we used not only Theorem~\ref{thm:one_non_flow_Joshi_thm}, but also the rest of tools reviewed in Section~\ref{sec:multistationarity_methods}. 

We included all atoms of multistationarity, smallest embedded subnetworks that are multistationary, of bimolecular\footnote{meaning the sum of stoichiometric coefficients in every complex is at most two.} two non-flow reactions that are classified in \cite[Figure 3]{Joshi-Shiu-2013}. There are 35 of these. Then we created some extended networks from these atoms which are again multistationary due to Theorem~\ref{thm:inheritance}.

\texttt{CRNToolbox} was also used to generate some examples and label them. Unfortunately generation of networks in \texttt{CRNToolbox} can not be easily automated: this tool is designed for the user to enter every example manually. The software does not have an API interface to allow for calls from code.   We also observed that for networks with large complexes or many reactions the software does not work well.

As seen in Section~\ref{sec:example_for_limitation_of_4_multi_methods} the algebraic approach also takes too much time and memory and thus was deemed not practical for dataset generation.

Using Theorem~\ref{thm:inheritance} we were able to create many multistationary networks from any other multistationary case that we could find. However, creating non-multistationary examples is harder and the dataset up to this step is very unbalanced. That led us to the positive feedback loop approach and Theorem~\ref{thm:positive_feedback_loops}.

To create a sufficient number of non-multistationary networks to make our dataset balanced we first identified 45 fully open networks with no positive feedback loops as listed in Appendix~\ref{sec:appendix_non_multi_classes_using_feedback_loop_results}. Then, by changing the stoichiometric coefficients of these networks following Corollary~\ref{cor:infinite_non_multi}, we generated new non-multistationary examples to add to the dataset.

The final dataset consists of around 104,000 CRNs, split evenly into two classes of multistationary and non-multistationary. These are created using three tools: 5,621 multistationary and 6,898 non-multistationary networks are generated using Theorem~\ref{thm:one_non_flow_Joshi_thm}, 46,499 multistationary networks are created by Theorem~\ref{thm:inheritance} using the atoms of multistationary given at \cite[Figure 3]{Joshi-Shiu-2013}, and finally 45,410 non-multistationary networks are created using Proposition~\ref{prop:changing_sc_without_loosing_NPFL} and the 45 base networks with no positive feedback loops in Appendix~\ref{sec:appendix_non_multi_classes_using_feedback_loop_results}. The number of species, number of reactions, and the individual stoichiometric coefficients are set to vary in the ranges $2-5$, $1-5$, and $0-5$ respectively, inclusive of the end points of the ranges. Table~\ref{tab:dataset_stats} shows the distribution of the networks in this dataset with respect to the number of species and reactions.

\begin{table}[h!]
	\centering
	\caption{The split of the dataset with respect to the number of species and reactions involved in its networks. \label{tab:dataset_stats} }
	\begin{tabular}{|c|c|c|c|c|c||c|}
		\hline
		no. species $\backslash$ no. reactions & 1 & 2 & 3 & 4 & 5 & Sum \\
		\hline
		2 & 1053 & 420 & 0 & 0 & 0 & 1473 \\
		\hline
		3 & 3503 & 3413 & 7641 & 4500 & 4500 & 23557 \\
		\hline
		4 & 3966 & 3790 & 11899 & 14287 & 6600 & 40542 \\
		\hline
		5 & 3997 & 5236 & 10935 & 12088 & 6600 & 38856 \\
		\hline
		\hline
		Sum & 12519 & 12859 & 30475 & 30875 & 17700 & 104428 \\
		\hline
	\end{tabular}
\end{table}

For validation purpose, we created a separate 32 additional networks manually, keeping the numbers of reactions and species in the range $2-6$.  None of these networks belong to the dataset above; further, we can prove that none of the non-multistationary ones could possibly be written as the result of a negative friendly operation from the 45 networks in Appendix~\ref{sec:appendix_non_multi_classes_using_feedback_loop_results}. These networks could all be labelled using \texttt{CRNToolbox} and they are listed in Appendix~\ref{sec:appendix_validating ntworks}.  Since these represent fundamentally different networks to those in the training set they will be used for validation $-$ to evaluate the level of generalisability of the ML model.

\subsection{ML experiment methodology}

Most GNNs do not consider directions of the graph edges, instead treating every edge similarly, which does not fit our application in which the directions are meaningful.  This led us to work with \emph{graph attention networks} (GATs\footnote{The acronym GAT is chosen to distinguish this from generative adversarial networks: another tool in the ML family which is usually referred to with the acronym GAN.}) as first presented in \cite{GAT-paper-2018}, which crucially for us allow for directed edges.  

The main distinction of GATs over GNNs is the use of a self-attention mechanism (inspired by transformers \cite{selfattention}).  In ML, an attention mechanism is a technique that allows models to focus on specific parts of the input data when producing outputs, allowing for variable sized inputs, with a self-attention mechanism the use of these to compute a representation of single sequence.  The architecture of GATs allows training to efficiently specify weights to different nodes in a neighbourhood without knowing the graph structure upfront.  Thus GATs usually perform better than GNNs in cases where the graph contains nodes or edges of varying importance to a prediction.  We hypothesise that this will be the case for our problem: the architecture may allow for learning to focus on parts of the network that promote or block multistationarity. The disadvantage of GATs compared to GNNs is a higher computational cost, however, this extra cost is not significant for the size of experiments presented here.

We use the GAT implementation in DGL \cite{Wang-et-al-2019-DGL-library} running on  PyTorch \cite{Paszke-et-al-2019-Pytorch}. We use the cross-entropy loss function\footnote{The cross-entropy loss function measures the difference between a model's predicted probability distribution and the true probability distribution in classification tasks. It is calculated as the negative sum of the true probabilities multiplied by the logarithm of the predicted probabilities for each class.} with the input, hidden, output, projection dimensions, batch size, and number of epochs set to 1, 4, 3, 3, 128, 100 respectively.

\subsection{Results}

Figure \ref{fig:GAT_loss} shows the training loss (the aggregated difference between a model's predicted output and the actual target values) and accuracy of classifications of the GAT model as it progressed through the training process.  In both cases we observe a fairly swift and smooth convergence.

\begin{figure}[ht]
	\begin{center}
		\includegraphics[width=5.5cm]{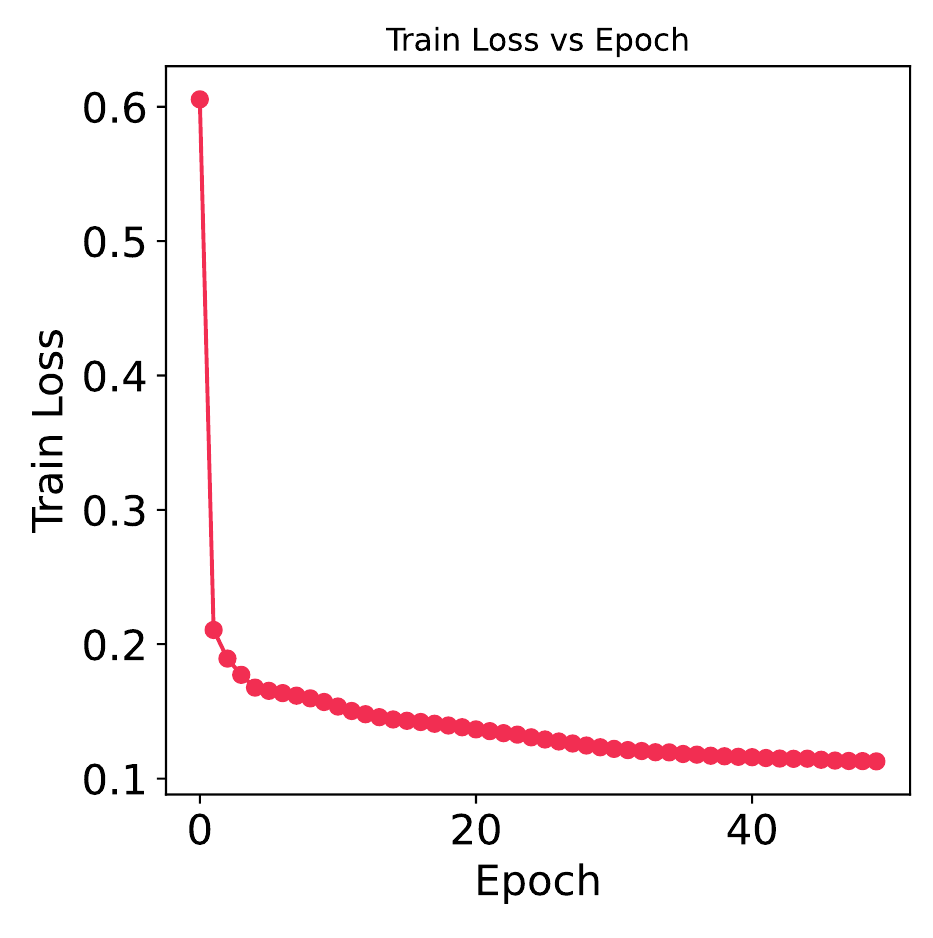}
		\includegraphics[width=5.5cm]{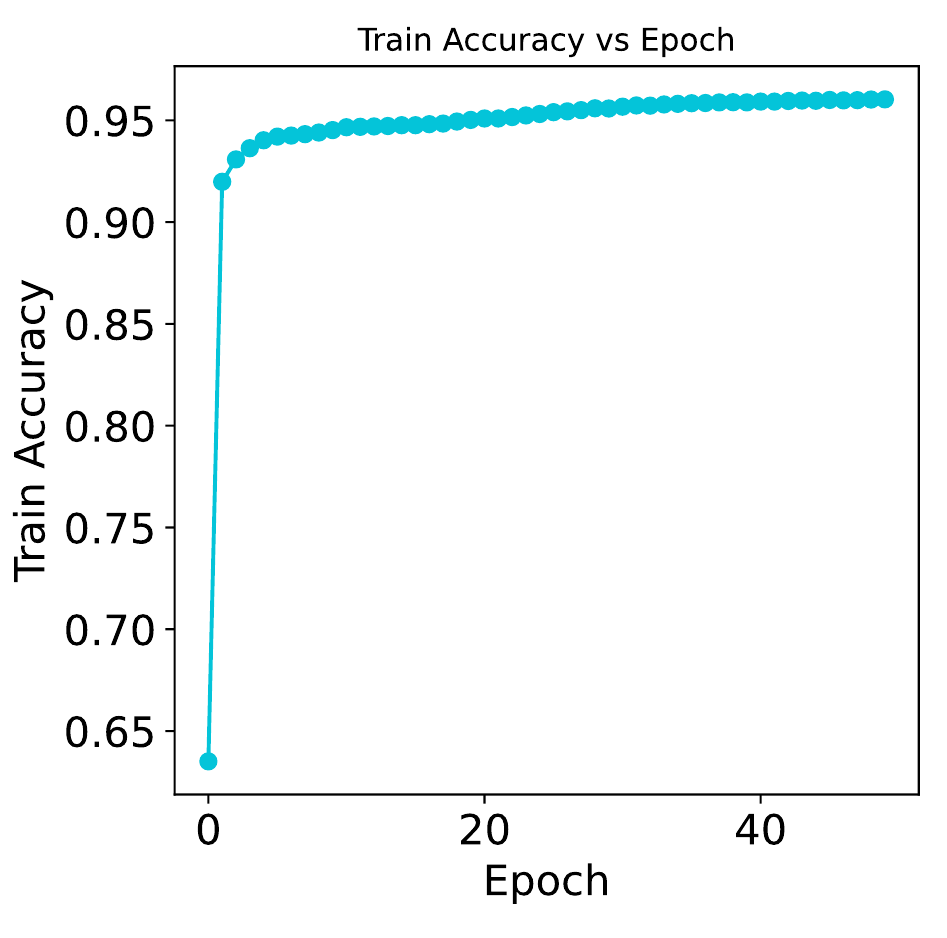}
	\end{center}
	\caption{The training loss and accuracy of the GAT model throughout the training process.}
	\label{fig:GAT_loss}
\end{figure}

The trained GAT model achieves 96\% accuracy on both the training and testing datasets (separate subsets of the original dataset described).  We can conclude that the new graph representation allows for successful machine learning, in particular without the limitations of simpler models described in Section \ref{sec:ML}.

The model outputs the binary classification, but prior to this, the model can project the input CRN to three dimensions (as a subsequent step before projecting further to that final decision).  The advantage of this is that the 3-dimensional data can be used to give a visualisation of its classification.  An example of this is within Figure \ref{fig:GAT_plots}, showing that under this projection, the labelled instances are very well separated.

\begin{figure}[ht]
	\centering 	
	\includegraphics[width=9.00cm]{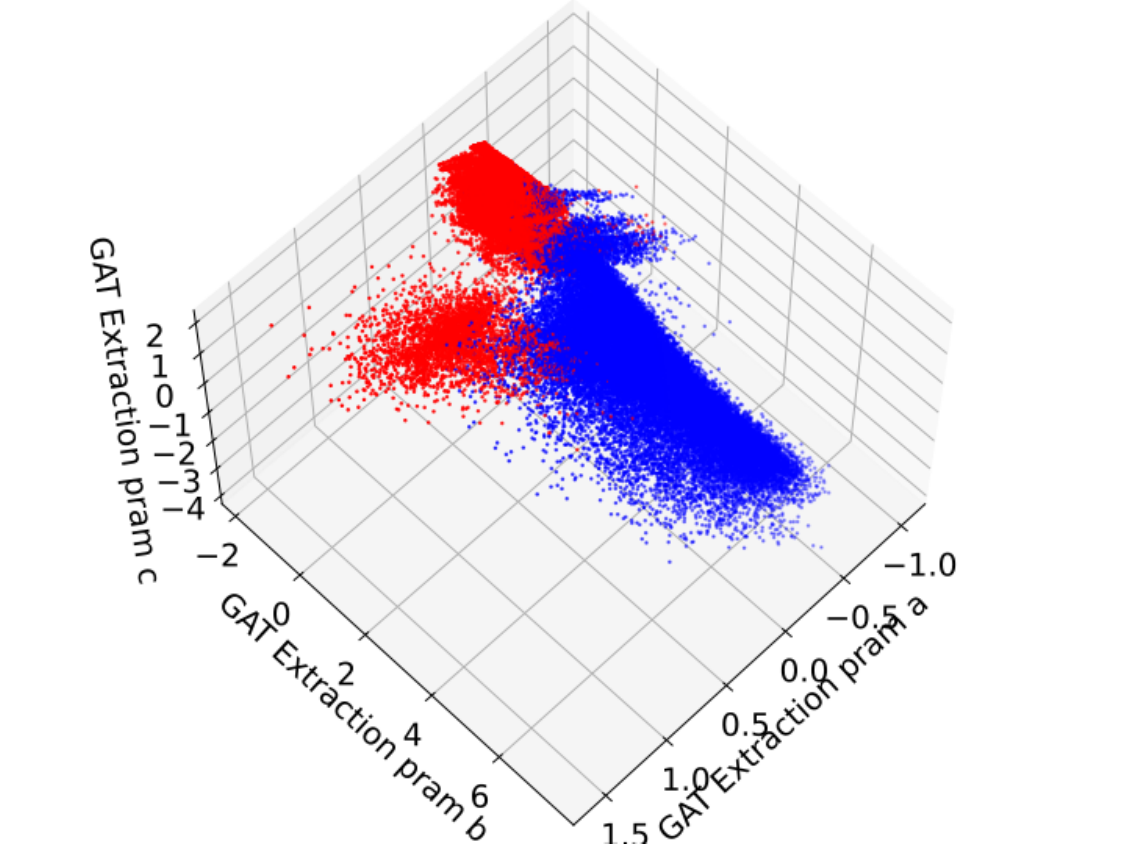}
	\caption{Visualization of the class separation provided by the GAT model on the training data.}
	\label{fig:GAT_plots}
\end{figure}

\subsection{Validation on independent data}

A key question of any ML experiment is whether the model can generalise away from the training data.  Although we tested and saw good predictive performance on data not used in training, that testing data came from the same original source dataset and thus will have many similarities to the training data.  It is desirable to also validate performance on independently sourced data.

To explore this we use the separate validation set of 32 examples described at the end of Section~\ref{subsec:dataset} and itemised in Appendix~\ref{sec:appendix_validating ntworks}. These examples were produced completely independently of the main dataset. Table~\ref{tab:confusion_matrix}(a) shows the confusion matrix for the performance of the GAT on this data: the instances in the shaded cells are correctly predicted and the others incorrect. The overall accuracy is 72\%, and this is observed both in predictions of multistationarity and non-multistationarity. While there is still much room for improvement, the model is demonstrating generalisation, i.e. performance on independently sourced data.

\begin{table}[h!]
	\centering
	\caption{\label{tab:confusion_matrix} Confusion Matrix of performance of the trained GAT model on (a) the 32 networks shown in Appendix~\ref{sec:appendix_validating ntworks}; and on (b) the 62 fully open 2-CRN selected from dataset of \cite{Banaji-2017} as explained in the text (and available in the code release supporting this paper).}
	\begin{tabular}[h!]{cc}
		(a)
		&
		(b)
		\vspace{0.25cm}
		\\
		\begin{minipage}{0.5\linewidth}
			\resizebox{0.8\textwidth}{!}{
				\begin{tabular}{|c|c|c|}
					\hline
					\multirow{2}{*}{\textbf{Actual}} & \multicolumn{2}{c|}{\textbf{Predicted}} \\
					\cline{2-3}
					& \textbf{Multi} & \textbf{Non-Multi} \\
					\hline
					\textbf{Multi} & \cellcolor{YellowGreen} 8 & 3 \\
					\hline
					\textbf{Non-Multi} & 6 & \cellcolor{YellowGreen} 15 \\
					\hline
				\end{tabular}
			}
		\end{minipage}
		&
		\begin{minipage}{0.5\linewidth}
			\resizebox{0.8\textwidth}{!}{
				\begin{tabular}{|c|c|c|}
					\hline
					\multirow{2.6}{*}{\textbf{Actual}} & \multicolumn{2}{c|}{\textbf{Predicted}} \\
					\cline{2-3}
					& \textbf{Multi} & \textbf{Non-Multi} \\
					\hline
					\textbf{Multi} & \cellcolor{YellowGreen} 2 & 3 \\
					\hline
					\textbf{Non-Multi} & 22 & \cellcolor{YellowGreen} 35 \\
					\hline
				\end{tabular}
			}
		\end{minipage}
	\end{tabular}
\end{table}

We encourage readers to perform further validation on their models of interest.  At the code of the paper we provided a Jupyter notebook called \texttt{MLCRN\_Presentation\_Code.ipynb} where the user can feed a new fully open network and ask the prediction of the trained model for its multistationarity. We also provided the possibility to feed a collection of fully open networks together with their true multistationarity labels in order to get the confusion matrix of the performance of the model on the user's new data. 

\subsection{Further validation with NAUTY}

One further source of validation we considered is via \cite{Banaji-2017}. The author of \cite{Banaji-2017} used the graph isomorphism software NAUTY \cite{MP14} to enumerate all non-isomorphic networks with complexes taken from the set of small shapes $\{ 0, X, 2X, X+Y\}$, and with number of species and reactions constrained to at most six.  They provide an open dataset of such networks, which are called \emph{2-CRNs}. The dataset has $218,488,063$ fully open 2-CRNs at the time of writing. Although the size is suitable for training a graph learning algorithm, the data does not include labels indicating the multistationarity of the networks. Therefore instead of using it for training a model, we selected some of the networks from this dataset to manually label and use for another validation test of our trained model. 

The data from \cite{Banaji-2017} is stored in different classes depending on the number of species and reactions. From each class we selected the first, the last and the middle network: so we use at most three from each class\footnote{\emph{at most}, because not every class of networks had three to chose from: for example, there was only one network in class with 1 species and 4 reactions.}. We also avoided classes where the file size was more than 100 MB. This selection process identified 67 networks. These networks and the code that selected them systematically can be found in the code release that supports the present paper. 

From these 67 networks, we find 61 can be checked for multistationarity by CRNToolbox. For the other 6 networks, the tests using CRNToolbox either terminate with an error message or declare an inconclusive result. From these 6 networks, we were able to conclude multistationarity of one of them since it contained an atom of multistationarity (see Section~\ref{sec:multi_and_inheritance}) and we used the implementation of the algebraic method in Maple (see Section~\ref{sec:multi_and_algebra}) to conclude non-multistationarity for another 3. There remains two networks of the 67 that none of the methods reviewed in Section~\ref{sec:multistationarity_methods} could tackle. Therefore we have 65 labelled networks. However, because our model works only for an input network that has a connected SRSC graph, three of these networks were found not suitable. In the end for this validation we have 62 acceptable fully open 2-CRN labelled as multistationary and another 5 labelled as non-multistationary. 

The confusion matrix of the performance of our trained model on this collection is given in Table~\ref{tab:confusion_matrix}(b). For these networks, the model does not perform as well as with the original testing data or the previous validation dataset.  The overall accuracy is 60\%, with 61\% accuracy on the non-multistationary class and 40\% on the multistationary class. A likely reason for this weaker performance is that while the number of species and reactions in the training data is similar to this 2-CRN dataset (in the range 1 to 6), that data also had freedom for the stoichiometric coefficients to be more than two, and to have complexes may have more than two species. Therefore the model is not particularly well specialised on finding the patterns specific to 2-CRNs.

\section{Conclusions and future work}
\label{sec:conclusions}

Our new theoretical results on multistationarity have allowed us to generate a new large scale dataset and demonstrate the feasibility of ML to predict properties of CRNs.  This can offer a powerful alternative or complementary approach to the existing methods to uncover this property.  The ML approach pays a high computational cost once (during ML training) but is cheap to make predictions with thereafter, compared to existing methods with a substantial cost per prediction.  Of course, these ML predictions are not always correct, and so researchers may choose to pay the high cost to verify the result: such a cost could be chosen selectively though, e.g. just for the model selected to study in detail.

We suggest that graph learning is the natural ML methodology for answering such questions about CRNs and have developed a graph representation of CRNs suitable for the purpose.  This approach removes the limitation of having fixed length input and arguably better represents the underlying networks than other embeddings for ML.  We note that the representation and methodology of Section \ref{sec:GAT} could be applied for other classifications in CRN with relative ease.

Future work could of course experiment with other representations and learning algorithms.  However, we suspect that future improvements to generalisability will have to come from a more diverse dataset.  We have contributed the first large scale labelled dataset in this field (available online at the URL at the end).  However, this does not yet sufficiently represent the space of CRNs.  Thus we need need new theoretical results to guide the data generation.

We finish by speculating on another avenue of future work: to investigate the possibility that the ML models trained for such tasks might inform new mathematical results.  The use of ML for mathematical discovery has been investigated for example in \cite{DVBBZTTBBJLWHK21, He2022}.  A route to this may come from applying the burgeoning literature on Explainable / Interpretable AI methods \cite{BDDBTBGGMBCH20}.  For example, recently an XAI analysis of ML predictions to optimise a computer algebra algorithm led to the creation of new simple heuristics that are not data-dependent and can be easily interpreted and implemented by humans \cite{PdREC24}.  We propose that such analysis might lead mathematicians to new CRN results on multistationarity.

\section*{Data Access Statement}

The code and data described in this paper is openly available from this URL: \url{https://doi.org/10.5281/zenodo.17054846}

\section*{Acknowledgements}

The first author acknowledges funding from Coventry University for a summer research internship.  The second and third authors  acknowledge the support of UKRI EPSRC Grant EP/T015748/1, ``\textit{Pushing Back the Doubly-Exponential Wall of Cylindrical Algebraic Decomposition} (the DEWCAD Project).

\FloatBarrier
\newpage


\appendix

\section{Fully open networks with no positive feedback loop}
\label{sec:appendix_non_multi_classes_using_feedback_loop_results}

In this Appendix we list 45 fully open networks with no positive feedback loops, thus by Theorem~\ref{thm:positive_feedback_loops} they are not multistationary. These networks are used as bases to generate more non-multistationary networks following Proposition~\ref{prop:changing_sc_without_loosing_NPFL} for the training dataset of GAT in Section~\ref{sec:GAT}. 

We give the full DSR-graph for just one of the networks, Network~\eqref{eq:NPFL_10}, below. Note that the vertices of species inflow reactions have degree one so they can not be in any loop. The nodes of species outflow reactions can only be in the trivial negative loop of length one. Therefore, since species flow reactions can not be used to create positive feedback loops, we omit the vertices and edges corresponding to these reactions in the rest of this section. We also switch to the signed DSR-graph and instead of putting labels $+$ and $-$, we colour the negative edges with blue and keep positive edges as black.  

\noindent For the other networks the reactions $X_i\ce{<=>[k_{2i+r}][k_{2i+r-1}]}0$, $i=1,\dots,n$ where $n$ is the number of species and $r$ is the number of non-species-flow reactions are also omitted to save space, but the reader should remember that all these networks are fully open and thus these reactions are actually present.

\begin{figure}[ht!]
	\centering

	\caption{The full DSR-graph of Network~\eqref{eq:NPFL_10} together with the network. \label{fig:NPFL_full_graph}}
\end{figure}

\begin{figure}[ht!]
	\label{fig:NPFL_1_to_7}
	\centering
	%
		\end{array}
	\end{equation}

\end{figure}

\FloatBarrier


\section{Validating examples}
\label{sec:appendix_validating ntworks}

This appendix contains the 32 validating CRN examples used in Section~\ref{sec:GAT}. The first 21 networks, namely \eqref{eq:Val1}$-$\eqref{eq:Val21}, are non-multistationary and the remaining 11, \eqref{eq:Val22}$-$\eqref{eq:Val32}, are multistationary. Similar to Appendix~\ref{sec:appendix_non_multi_classes_using_feedback_loop_results}, the species flows reactions are omitted.

\begin{figure}[ht!]
	\centering
	\begin{tabular}{cc}
		\begin{minipage}{0.5\linewidth}
			\begin{equation}
				\label{eq:Val1}
				\begin{array}{c}
					X_1 \ce{->[k_1]} X_2\\
					X_2 \ce{->[k_2]} 2X_1\\
				\end{array}
			\end{equation}
		\end{minipage}
		&
		\begin{minipage}{0.5\linewidth}
			\begin{equation}
				\label{eq:Val2}
				\begin{array}{c}
					2X_1 \ce{->[k_1]} 2X_2\\
					2X_2 \ce{->[k_2]} X_1\\
				\end{array}
			\end{equation}
		\end{minipage}
	\end{tabular}
\end{figure}

\begin{figure}[ht!]
	\centering
	\begin{tabular}{cc}
		\begin{minipage}{0.5\linewidth}
			\begin{equation}
				\label{eq:Val3}
				\begin{array}{c}
					X_1 \ce{->[k_1]} X_2\\
					2X_1 \ce{->[k_2]} X_2\\
				\end{array}
			\end{equation}
		\end{minipage}
		&
		\begin{minipage}{0.5\linewidth}
			\begin{equation}
				\label{eq:Val4}
				\begin{array}{c}
					4X_1 \ce{->[k_1]} X_2\\
					X_2 \ce{->[k_2]} 2X_1\\
				\end{array}
			\end{equation}
		\end{minipage}
	\end{tabular}
\end{figure}

\begin{figure}[ht!]
	\centering
	\begin{tabular}{cc}
		\begin{minipage}{0.5\linewidth}
			\begin{equation}
				\label{eq:Val5}
				\begin{array}{c}
					X_1 \ce{->[k_1]} X_2\\
					X_1 + X_2 \ce{->[k_2]} 2X_3\\
				\end{array}
			\end{equation}
		\end{minipage}
		&
		\begin{minipage}{0.5\linewidth}
			\begin{equation}
				\label{eq:Val6}
				\begin{array}{c}
					X_1 \ce{->[k_1]} 2X_2\\
					X_1 + X_2 \ce{->[k_2]} 2X_3\\
				\end{array}
			\end{equation}
		\end{minipage}
	\end{tabular}
\end{figure}

\begin{figure}[ht!]
	\centering
	\begin{tabular}{cc}
		\begin{minipage}{0.5\linewidth}
			\begin{equation}
				\label{eq:Val7}
				\begin{array}{c}
					X_1 + X_2 \ce{->[k_1]} X_3\\
					X_2 + X_3 \ce{->[k_2]} X_1\\
				\end{array}
			\end{equation}
		\end{minipage}
		&
		\begin{minipage}{0.5\linewidth}
			\begin{equation}
				\label{eq:Val8}
				\begin{array}{c}
					X_1 \ce{<=>[k_1][k_2]} 2X_2\\
					X_1 + X_2 \ce{->[k_3]} 2X_3
				\end{array}
			\end{equation}
		\end{minipage}
	\end{tabular}
\end{figure}

\begin{figure}[ht!]
	\centering
	\begin{tabular}{cc}
		\begin{minipage}{0.5\linewidth}
			\begin{equation}
				\label{eq:Val9}
				\begin{array}{c}
					X_1 \ce{->[k_1]} 2X_2\\
					X_1 + X_2 \ce{<=>[k_2][k_3]} 2X_3
				\end{array}
			\end{equation}
		\end{minipage}
		&
		\begin{minipage}{0.5\linewidth}
			\begin{equation}
				\label{eq:Val10}
				\begin{array}{c}
					X_1 + X_2 \ce{->[k_1]} 2X_3\\
					X_1 \ce{->[k_2]} 2X_1 + X_2\\
					X_3 \ce{->[k_3]} 2X_2\\
				\end{array}
			\end{equation}
		\end{minipage}
	\end{tabular}
\end{figure}

\begin{figure}[ht!]
	\centering
	\begin{tabular}{cc}
		\begin{minipage}{0.5\linewidth}
			\begin{equation}
				\label{eq:Val11}
				\begin{array}{c}
					X_1 \ce{<=>[k_1][k_2]} 2X_2\\
					X_1 + X_2 \ce{<=>[k_3][k_4]} 2X_3
				\end{array}
			\end{equation}
		\end{minipage}
		&
		\begin{minipage}{0.5\linewidth}
			\begin{equation}
				\label{eq:Val12}
				\begin{array}{c}
					X_1 \ce{->[k_1]} X_2\\
					X_1 \ce{->[k_2]} X_3\\
					X_1 + X_2 + X_3 \ce{->[k_3]} X_1\\
					X_2 + X_3 \ce{->[k_4]} 0\\
				\end{array}
			\end{equation}
		\end{minipage}
	\end{tabular}
\end{figure}

\begin{figure}[ht!]
	\centering
	\begin{tabular}{cc}
		\begin{minipage}{0.5\linewidth}
			\begin{equation}
				\label{eq:Val13}
				\begin{array}{c}
					X_1 \ce{->[k_1]} X_2 + X_3\\
					X_2 + X_3 + X_4 \ce{->[k_2]} 0\\
					X_1 + X_4 \ce{->[k_3]} 0\\
				\end{array}
			\end{equation}
		\end{minipage}
		&
		\begin{minipage}{0.5\linewidth}
			\begin{equation}
				\label{eq:Val14}
				\begin{array}{c}
					X_1 + X_2 \ce{->[k_1]} X_3\\
					X_2 + X_3 \ce{->[k_2]} X_4\\
					X_1 + X_4 \ce{->[k_3]} 0\\
				\end{array}
			\end{equation}
		\end{minipage}
	\end{tabular}
\end{figure}

\begin{figure}[ht!]
	\centering
	\begin{tabular}{cc}
		\begin{minipage}{0.5\linewidth}
			\begin{equation}
				\label{eq:Val15}
				\begin{array}{c}
					X_1 \ce{->[k_1]} X_2\\
					X_2 \ce{->[k_2]} X_3\\
					X_3 \ce{->[k_3]} X_4\\
					X_4 \ce{->[k_4]} X_1\\
				\end{array}
			\end{equation}
		\end{minipage}
		&
		\begin{minipage}{0.5\linewidth}
			\begin{equation}
				\label{eq:Val16}
				\begin{array}{c}
					X_1 \ce{->[k_1]} X_2\\
					X_2 \ce{->[k_2]} X_3\\
					X_3 \ce{->[k_3]} X_4\\
					X_1 + X_4 \ce{->[k_4]} X_4\\
				\end{array}
			\end{equation}
		\end{minipage}
	\end{tabular}
\end{figure}

\begin{figure}[ht!]
	\centering
	\begin{tabular}{cc}
		\begin{minipage}{0.5\linewidth}
			\begin{equation}
				\label{eq:Val17}
				\begin{array}{c}
					X_1 \ce{->[k_1]} X_2\\
					X_2 \ce{->[k_2]} X_3\\
					X_1 + X_4 \ce{->[k_3]} X_1\\
					X_1 + X_3 \ce{->[k_4]} X_4\\
				\end{array}
			\end{equation}
		\end{minipage}
		&
		\begin{minipage}{0.5\linewidth}
			\begin{equation}
				\label{eq:Val18}
				\begin{array}{c}
					X_1 \ce{->[k_1]} X_3\\
					X_1 + X_2 \ce{->[k_2]} X_2 + X_3\\
					X_1 + 2X_2 \ce{->[k_3]} 2X_2 + X_3\\
					X_1 + 3X_2 \ce{->[k_4]} 3X_2 + X_3\\
					X_2 \ce{->[k_5]} X_4\\
				\end{array}
			\end{equation}
		\end{minipage}
	\end{tabular}
\end{figure}

\begin{figure}[ht!]
	\centering
	\begin{tabular}{cc}
		\begin{minipage}{0.5\linewidth}
			\begin{equation}
				\label{eq:Val19}
				\begin{array}{c}
					X_1 \ce{->[k_1]} X_2 + X_3 + X_4\\
					X_3 + X_4 \ce{->[k_2]} X_5\\
					X_2 + X_5 \ce{->[k_3]} X_2\\
				\end{array}
			\end{equation}
		\end{minipage}
		&
		\begin{minipage}{0.5\linewidth}
			\begin{equation}
				\label{eq:Val20}
				\begin{array}{c}
					X_1 + X_2 + X_3 \ce{->[k_1]} X_3 + X_4 + X_5\\
					X_2 + X_4 \ce{->[k_2]} X_3\\
					X_1 + 2X_5 \ce{->[k_3]} X_1\\
				\end{array}
			\end{equation}
		\end{minipage}
	\end{tabular}
\end{figure}

\begin{figure}[ht!]
	\centering
	\begin{tabular}{cc}
		\begin{minipage}{0.5\linewidth}
			\begin{equation}
				\label{eq:Val21}
				\begin{array}{c}
					X_1 + X_2 \ce{->[k_1]} X_3\\
					X_4 + X_5 \ce{->[k_2]} X_6\\
					X_3 + X_6 \ce{->[k_3]} X_1 + X_6\\
				\end{array}
			\end{equation}
		\end{minipage}
		&
		\begin{minipage}{0.5\linewidth}
			\begin{equation}
				\label{eq:Val22}
				\begin{array}{c}
					X_1 \ce{->[k_1]} X_2\\
					X_1 + X_2 \ce{->[k_2]} 2X_3\\
					2X_3 \ce{->[k_3]} 2X_1
				\end{array}
			\end{equation}
		\end{minipage}
	\end{tabular}
\end{figure}

\begin{figure}[ht!]
	\centering
	\begin{tabular}{cc}
		\begin{minipage}{0.5\linewidth}
			\begin{equation}
				\label{eq:Val23}
				\begin{array}{c}
					X_1 + X_2 \ce{->[k_1]} X_3 + X_4\\
					X_1 + X_4 \ce{->[k_2]} 2X_1\\
				\end{array}
			\end{equation}
		\end{minipage}
		&
		\begin{minipage}{0.5\linewidth}
			\begin{equation}
				\label{eq:Val24}
				\begin{array}{c}
					X_1 + X_2 \ce{->[k_1]} X_3 + X_4\\
					X_1 + X_3 \ce{->[k_2]} X_1 + X_4\\
					X_3 + X_4 \ce{->[k_3]} 2X_1\\
				\end{array}
			\end{equation}
		\end{minipage}
	\end{tabular}
\end{figure}

\begin{figure}[ht!]
	\centering
	\begin{tabular}{cc}
		\begin{minipage}{0.5\linewidth}
			\begin{equation}
				\label{eq:Val25}
				\begin{array}{c}
					X_1 + X_2 \ce{->[k_1]} X_1 + X_3\\
					2X_3 + X_4 \ce{->[k_2]} X_1 + X_2\\
					X_1 \ce{->[k_3]} X_4\\
				\end{array}
			\end{equation}
		\end{minipage}
		&
		\begin{minipage}{0.5\linewidth}
			\begin{equation}
				\label{eq:Val26}
				\begin{array}{c}
					X_1 \ce{->[k_1]} X_2\\
					X_2 \ce{->[k_2]} X_3\\
					X_1 + X_3 \ce{->[k_3]} 2X_4\\
					2X_4 \ce{->[k_4]} 2X_1\\
				\end{array}
			\end{equation}
		\end{minipage}
	\end{tabular}
\end{figure}

\begin{figure}[ht!]
	\centering
	\begin{tabular}{cc}
		\begin{minipage}{0.5\linewidth}
			\begin{equation}
				\label{eq:Val27}
				\begin{array}{c}
					X_1 \ce{->[k_1]} X_2 + X_3 + X_4\\
					X_3 + X_4 \ce{->[k_2]} X_5\\
					X_2 + X_5 \ce{->[k_3]} 2X_1\\
				\end{array}
			\end{equation}
		\end{minipage}
		&
		\begin{minipage}{0.5\linewidth}
			\begin{equation}
				\label{eq:Val28}
				\begin{array}{c}
					X_1 + X_2 \ce{->[k_1]}X_3\\
					X_3 + 2X_4 \ce{->[k_2]} X_5\\
					X_2 + X_5 \ce{->[k_3]} 2X_4 + X_5\\
				\end{array}
			\end{equation}
		\end{minipage}
	\end{tabular}
\end{figure}

\begin{figure}[ht!]
	\centering
	\begin{tabular}{cc}
		\begin{minipage}{0.5\linewidth}
			\begin{equation}
				\label{eq:Val29}
				\begin{array}{c}
					X_1 + X_2 \ce{->[k_1]} X_2 + X_3\\
					X_3 \ce{->[k_2]} X_4\\
					2X_4 + X_5 \ce{->[k_3]} 2X_3\\
					X_2 + X_5 \ce{->[k_4]} 0\\
				\end{array}
			\end{equation}
		\end{minipage}
		&
		\begin{minipage}{0.5\linewidth}
			\begin{equation}
				\label{eq:Val30}
				\begin{array}{c}
					X_1 + X_2 \ce{->[k_1]} X_3\\
					X_1 + X_3 \ce{->[k_2]} X_4\\
					X_4 + X_5 \ce{->[k_3]} 2X_2\\
					X_3 + X_5 \ce{->[k_4]} 0\\
				\end{array}
			\end{equation}
		\end{minipage}
	\end{tabular}
\end{figure}

\begin{figure}[ht!]
	\centering
	\begin{tabular}{cc}
		\begin{minipage}{0.5\linewidth}
			\begin{equation}
				\label{eq:Val31}
				\begin{array}{c}
					X_1 \ce{->[k_1]} X_1 + X_2\\
					2X_2 \ce{->[k_2]} X_1\\
					X_2 + X_3 \ce{->[k_3]} X_4 + X_5\\
					X_4 + X_5 \ce{->[k_4]} 2X_3
				\end{array}
			\end{equation}
		\end{minipage}
		&
		\begin{minipage}{0.5\linewidth}
			\begin{equation}
				\label{eq:Val32}
				\begin{array}{c}
					X_1 + X_2 \ce{->[k_1]} X_3 + X_4\\
					X_2 + X_3 \ce{->[k_2]} X_5\\
					X_1 + X_4 \ce{->[k_3]} X_6\\
					2X_3 + X_5 \ce{->[k_4]} X_5\\
					2X_4 + X_6 \ce{->[k_5]} X_6\\
					X_5 + X_6 \ce{->[k_6]} 2X_1\\
				\end{array}
			\end{equation}
		\end{minipage}
	\end{tabular}
\end{figure}


{\small
	\begin{thebibliography}{56}
		\newcommand{\enquote}[1]{``#1''}
		\providecommand{\natexlab}[1]{#1}
		\providecommand{\url}[1]{\texttt{#1}}
		\providecommand{\urlprefix}{URL }
		\expandafter\ifx\csname urlstyle\endcsname\relax
		\providecommand{\doi}[1]{doi:\discretionary{}{}{}#1}\else
		\providecommand{\doi}{doi:\discretionary{}{}{}\begingroup
			\urlstyle{rm}\Url}\fi
		
		\bibitem{Abakuks-1982}
		Andris Abakuks.
		\newblock {\em Review of The Mathematical Theory of the Dynamics of Biological
			Populations {{II}}}.
		\newblock Journal of the Royal Statistical Society. Series A (General),
		145(4):512--512, 1982, DOI:
		\href{https://doi.org/10.2307/2982110}{10.2307/2982110}.
		
		\bibitem{Banaji-2017}
		Murad Banaji.
		\newblock {\em Counting chemical reaction networks with {NAUTY}}. 2017,
		\newblock arXiv: \href{https://arxiv.org/abs/1705.10820}{1705.10820}.
		
		\bibitem{Banaji-2018}
		Murad Banaji.
		\newblock {\em Inheritance of oscillation in chemical reaction networks}.
		\newblock Applied Mathematics and Computation, 325:191--209, 2018, DOI:
		\href{https://doi.org/10.1016/j.amc.2017.12.012}{10.1016/j.amc.2017.12.012}.
		
		\bibitem{Banaji_2023}
		Murad Banaji.
		\newblock {\em Splitting Reactions Preserves Nondegenerate Behaviors in Chemical Reaction Networks}.
		\newblock SIAM Journal on Applied Mathematics, 83(2):748--769, 2023, DOI:
		\href{https://doi.org/10.1137/22M1478392}{10.1137/22M1478392}.
		
		\bibitem{Banaji-Boros-Hofbauer-2023}
		Murad Banaji, Bal\'azs Boros, and Josef Hofbauer.
		\newblock {\em The inheritance of local bifurcations in mass action networks}. 2023,
		\newblock arXiv: \href{https://arxiv.org/abs/2312.12897}{2312.12897}.
		
		\bibitem{Banaji-Pantea-2018}
		Murad Banaji and Casian Pantea.
		\newblock {\em The Inheritance of Nondegenerate Multistationarity in Chemical
			Reaction Networks}.
		\newblock SIAM Journal on Applied Mathematics, 78(2):1105--1130, 2018, DOI:
		\href{https://doi.org/10.1137/16M1103506}{10.1137/16M1103506}.
		
		\bibitem{BDDBTBGGMBCH20}
		Alejandro {Barredo Arrieta}, Natalia {D\'{i}az-Rodr\'{i}guez}, Javier {Del
			Ser}, Adrien Bennetot, Siham Tabik, Alberto Barbado, Salvador Garcia, Sergio
		{Gil-Lopez}, Daniel Molina, Richard Benjamins, Raja Chatila, and Francisco
		Herrera.
		\newblock {\em Explainable Artificial Intelligence ({XAI}): Concepts,
			taxonomies, opportunities and challenges toward responsible {AI}}.
		\newblock Information Fusion, 58:82--115, 2020, DOI:
		\href{https://doi.org/10.1016/j.inffus.2019.12.012}{10.1016/j.inffus.2019.12.012}.
		
		\bibitem{Bradford-Davenport-England-Errami-et-al-2020}
		Russell Bradford, James~H. Davenport, Matthew England, Hassan Errami, Vladimir
		Gerdt, Dima Grigoriev, Charles Hoyt, Marek Ko\v{s}ta, Ovidiu Radulescu,
		Thomas Sturm, and Andreas Weber.
		\newblock {\em Identifying the parametric occurrence of multiple steady states
			for some biological networks}.
		\newblock Journal of Symbolic Computation, 98:84--119, 2020, DOI:
		\href{https://doi.org/10.1016/j.jsc.2019.07.008}{10.1016/j.jsc.2019.07.008}.
		
		\bibitem{Bradford-Davenport-England-McCallum-Wilson-2016}
		Russell Bradford, James~H. Davenport, Matthew England, Scott McCallum, and
		David Wilson.
		\newblock {\em Truth Table Invariant Cylindrical Algebraic Decomposition}.
		\newblock Journal of Symbolic Computation, 76:1--35, 2016, DOI:
		\href{https://doi.org/10.1016/j.jsc.2015.11.002}{10.1016/j.jsc.2015.11.002}.
		
		\bibitem{Buchberger2006}
		Bruno Buchberger.
		\newblock {\em Bruno {B}uchberger's {PhD} thesis (1965): {A}n algorithm for
			finding the basis elements of the residue class ring of a zero dimensional
			polynomial ideal}.
		\newblock Journal of Symbolic Computation, 41(3-4):475--511, 2006, DOI:
		\href{https://doi.org/10.1016/j.jsc.2005.09.007}{10.1016/j.jsc.2005.09.007}.
		
		\bibitem{Collins1998b}
		George~E. Collins.
		\newblock {\em Quantifier Elimination for Real Closed Fields by Cylindrical
			Algebraic Decomposition}.
		\newblock In Quantifier Elimination and Cylindrical Algebraic Decomposition,
		Texts \& Monographs in Symbolic Computation, pages 85--121. Springer-Verlag,
		1998, DOI: \href{https://doi.org/10.1007/978-3-7091-9459-1_4}{10.1007/978-3-7091-9459-1\_4}.
		
		\bibitem{Cox-et-al-undergraduate-2015}
		David A.~Cox, John~Little and Donal~O'Shea.
		\newblock {\em Ideals, Varieties, and Algorithms}.
		\newblock Springer Cham, 2015, DOI: \href{https://doi.org/10.1007/978-3-319-16721-3}{10.1007/978-3-319-16721-3}.
		
		\bibitem{Hungarian_Lemma_2020}
		Gheorghe Craciun, Matthew D. Johnston, G\'abor Szederk\'enyi, Elisa Tonello, J\'anos T\'oth and Polly Y. Yu.
		\newblock {\em Realizations of kinetic differential equations}.
		\newblock Mathematical Biosciences and Engineering, 17(1):862--892 2020, DOI: \href{https://doi.org/10.3934/mbe.2020046}{10.3934/mbe.2020046}.
		
		\bibitem{Cristianini-Shawe_Taylor-2000}
		Nello Cristianini and John Shawe-Taylor.
		\newblock {\em An Introduction to Support Vector Machines and Other
			Kernel-based Learning Methods}.
		\newblock Cambridge University Press, 2000, DOI: \href{https://doi.org/10.1017/CBO9780511801389}{10.1017/CBO9780511801389}.
		
		\bibitem{DVBBZTTBBJLWHK21}
		Alex Davies, Petar Veli\v{c}kovi\'{c}, Lars Buesing, Sam Blackwell, Daniel
		Zheng, Nenad Toma\v{s}ev, Richard Tanburn, Peter Battaglia, Charles Blundell,
		Andr\'{a}s Juh\'{a}sz, Marc Lackenby, Geordie Williamson, Demis Hassabis, and
		Pushmeet Kohli.
		\newblock {\em Advancing mathematics by guiding human intuition with {AI}}.
		\newblock Nature, 600:70--74, 2021, DOI:
		\href{https://doi.org/10.1038/s41586-021-04086-x}{10.1038/s41586-021-04086-x}.
		
		\bibitem{Diperna-Lions-1989}
		R.~J. DiPerna and P.~L. Lions.
		\newblock {\em On the {{Cauchy}} Problem for {{Boltzmann}} Equations:
			{{Global}} Existence and Weak Stability}.
		\newblock Annals of Mathematics, 130(2):321--366, 1989, DOI:
		\href{https://doi.org/10.2307/1971423}{10.2307/1971423}.
		
		\bibitem{Drexler-et-al-2019}
		D\'aniel~Andr\'as Drexler, Tam\'as Ferenci, Andr\'as F\"uredi, Gergely
		Szak\'acs, and Levente Kov\'acs.
		\newblock {\em Tumor dynamics modeling based on formalreaction kinetics}.
		\newblock Acta Polytechnica Hungarica, 16(10):31--44, 2019, DOI:
		\href{https://doi.org/10.12700/APH.16.10.2019.10.3}{10.12700/APH.16.10.2019.10.3}.
		
		\bibitem{Ellison-1998}
		Phillipp~Raymond Ellison.
		\newblock {\em The advanced deficiency algorithm and its applications to
			mechanism discrimination}.
		\newblock PhD thesis, The University of Rochester, 1998.
		
		\bibitem{Feinberg-1972}
		Martin Feinberg.
		\newblock {\em Complex Balancing in General Kinetic Systems}.
		\newblock Archive for Rational Mechanics and Analysis, 49(3):187--194, 1972,
		DOI: \href{https://doi.org/10.1007/BF00255665}{10.1007/BF00255665}.
		
		\bibitem{Feinberg-1987}
		Martin Feinberg.
		\newblock {\em Chemical Reaction Network Structure and the Stability of Complex
			Isothermal Reactors \textemdash{I}. {The} Deficiency Zero and Deficiency One
			Theorems}.
		\newblock Chemical Engineering Science, 42(10):2229--2268, 1987, DOI:
		\href{https://doi.org/10.1016/0009-2509(87)80099-4}{10.1016/0009-2509(87)80099-4}.
		
		\bibitem{Feinberg-1995}
		Martin Feinberg.
		\newblock {\em Multiple steady states for chemical reaction networks of
			deficiency one}.
		\newblock Archive for Rational Mechanics and Analysis, 132(4):371--406, 1995,
		DOI: \href{https://doi.org/10.1007/BF00375615}{10.1007/BF00375615}.
		
		\bibitem{Feinberg-2019}
		Martin Feinberg.
		\newblock {\em Foundations of Chemical Reaction Network Theory}.
		\newblock Springer Cham, 2019,
		DOI: \href{https://doi.org/10.1007/978-3-030-03858-8}{10.1007/978-3-030-03858-8}.
		
		\bibitem{CRNToolbox-2018}
		Martin Feinberg, Phillipp Ellison, Haixia Ji, and Daniel Knight.
		\newblock {\em The Chemical Reaction Network Toolbox, {Windows} Version}, 2018.
		
		\bibitem{Feliu-Wiuf-2013}
		Elisenda Feliu and Carsten Wiuf.
		\newblock {\em Simplifying biochemical models with intermediate species}.
		\newblock Journal of The Royal Society Interface, 10(87):20130484, 2013, DOI:
		\href{https://doi.org/10.1098/rsif.2013.0484}{10.1098/rsif.2013.0484}.
		
		\bibitem{Feliu-Wiuf-2015}
		Elisenda Feliu and Carsten Wiuf.
		\newblock {\em Finding the positive feedback loops underlying
			multi-stationarity}.
		\newblock BMC Systems Biology, 9(1):Article number 22, 2015, DOI:
		\href{https://doi.org/10.1186/s12918-015-0164-0}{10.1186/s12918-015-0164-0}.
		
		\bibitem{Gardner-Cantor-Collins-2000}
		Timothy~S. Gardner, Charles~R. Cantor, and James~J. Collins.
		\newblock {\em Construction of a genetic toggle switch in {E}scherichia coli}.
		\newblock Nature, 403(6767):339--342, 2000, DOI:
		\href{https://doi.org/10.1038/35002131}{10.1038/35002131}.
		
		\bibitem{RootFinding-package}
		J\"{u}rgen Gerhard, David~J. Jeffery, and Guillaume Moroz.
		\newblock {\em A package for solving parametric polynomial systems}.
		\newblock ACM Communications in Computer Algebra, 43(3/4):61--72, 2010, DOI:
		\href{https://doi.org/10.1145/1823931.1823933}{10.1145/1823931.1823933}.
		
		\bibitem{Guantes-Poyatos-2008}
		Ra\'ul Guantes and Juan~F. Poyatos.
		\newblock {\em Multistable Decision Switches for Flexible Control of Epigenetic
			Differentiation}.
		\newblock PLOS Computational Biology, 4(11):1--13, 2008, DOI:
		\href{https://doi.org/10.1371/journal.pcbi.1000235}{10.1371/journal.pcbi.1000235}.
		
		\bibitem{Haixia-2011}
		Haixia Ji.
		\newblock {\em Uniqueness of Equilibria for Complex Chemical Reaction
			Networks}.
		\newblock PhD thesis, Ohio State University, Department of Mathematics, 2011, \url{http://rave.ohiolink.edu/etdc/view?acc\_num=osu1307122057}.
		
		\bibitem{He2022}
		Yang-Hui He.
		\newblock {\em Machine-learning mathematical structures}.
		\newblock International Journal of Data Science in the Mathematical Sciences,
		1(1):1--25, 2022, DOI:
		\href{https://doi.org/10.1142/S2810939222500010}{10.1142/S2810939222500010}.
		
		\bibitem{Horn-Jackson-1972}
		F.~Horn and R.~Jackson.
		\newblock {\em General Mass Action Kinetics}.
		\newblock Archive for Rational Mechanics and Analysis, 47(2):81--116, 1972,
		DOI: \href{https://doi.org/10.1007/BF00251225}{10.1007/BF00251225}.
		
		\bibitem{Hudson-Mankin-1981}
		J.~L. Hudson and J.~C. Mankin.
		\newblock {\em Chaos in the {{Belousov}}\textendash{{Zhabotinskii}} Reaction}.
		\newblock The Journal of Chemical Physics, 74(11):6171--6177, 1981, DOI:
		\href{https://doi.org/10.1063/1.441007}{10.1063/1.441007}.
		
		\bibitem{Jirstrand-1995}
		Mats Jirstrand.
		\newblock {\em Cylindrical Algebraic Decomposition $-$ an Introduction}.
		\newblock {L}inkpings {U}niversity, 1995, \url{https://www.diva-portal.org/smash/get/diva2:315832/FULLTEXT02}.
		
		\bibitem{Joshi-2013}
		Badal Joshi.
		\newblock {\em Complete Characterization by Multistationarity of Fully Open
			Networks with One Non-Flow Reaction}.
		\newblock Applied Mathematics and Computation, 219(12):6931--6945, 2013, DOI:
		\href{https://doi.org/10.1016/j.amc.2013.01.027}{10.1016/j.amc.2013.01.027}.
		
		\bibitem{Joshi-Shiu-2013}
		Badal Joshi and Anne Shiu.
		\newblock {\em Atoms of multistationarity in chemical reaction networks}.
		\newblock Journal of Mathematical Chemistry, 51(1):153--178, 2013, DOI:
		\href{https://doi.org/10.1007/s10910-012-0072-0}{10.1007/s10910-012-0072-0}.
		
		\bibitem{Kothamachu-Feliu-Cardelli-Soyer-2015}
		Varun~B. Kothamachu, Elisenda Feliu, Luca Cardelli, and Orkun~S. Soyer.
		\newblock {\em Unlimited multistability and Boolean logic in microbial
			signalling}.
		\newblock Journal of The Royal Society Interface, 12(108):20150234, 2015, DOI:
		\href{https://doi.org/10.1098/rsif.2015.0234}{10.1098/rsif.2015.0234}.
		
		\bibitem{Lazard-Rouillier-2007}
		Daniel Lazard and Fabrice Rouillier.
		\newblock {\em Solving parametric polynomial systems}.
		\newblock Journal of Symbolic Computation, 42(6):636--667, 2007, DOI:
		\href{https://doi.org/10.1016/j.jsc.2007.01.007}{10.1016/j.jsc.2007.01.007}.
		
		\bibitem{Lichtblau-2021}
		Daniel Lichtblau.
		\newblock {\em Symbolic analysis of multiple steady states in a {MAPK} chemical
			reaction network}.
		\newblock Journal of Symbolic Computation, 105:118--144, 2021, DOI:
		\href{https://doi.org/10.1016/j.jsc.2020.06.004}{10.1016/j.jsc.2020.06.004}.
		
		\bibitem{Liu-Zhou-2020}
		Zhiyuan Liu and Jie Zhou.
		\newblock {\em Introduction to Graph Neural Networks}.
		\newblock Springer Cham, 2020, DOI:
		\href{https://doi.org/10.1007/978-3-031-01587-8}{10.1007/978-3-031-01587-8}.
		
		\bibitem{MP14}
		Brendan~D. McKay and Adolfo Piperno.
		\newblock {\em Practical graph isomorphism, {II}}.
		\newblock Journal of Symbolic Computation, 60:94--112,2014, DOI:
		\href{https://doi.org/10.1016/j.jsc.2013.09.003}{10.1016/j.jsc.2013.09.003}.
		
		\bibitem{Mayr-Meyer-1982}
		Ernst~W Mayr and Albert~R Meyer.
		\newblock {\em The Complexity of the Word Problems for Commutative Semigroups
			and Polynomial Ideals}.
		\newblock Advances in Mathematics, 46(3):305--329, 1982, DOI:
		\href{https://doi.org/10.1016/0001-8708(82)90048-2}{10.1016/0001-8708(82)90048-2}.
		
		\bibitem{Mayr-Ritscher-2013}
		Ernst~W. Mayr and Stephan Ritscher.
		\newblock {\em Dimension-Dependent Bounds for {{Gr\"obner}} Bases of Polynomial
			Ideals}.
		\newblock Journal of Symbolic Computation, 49:78--94, 2013, DOI:
		\href{https://doi.org/10.1016/j.jsc.2011.12.018}{10.1016/j.jsc.2011.12.018}.
		
		\bibitem{Moroz-PhD-thesis}
		Guillaume Moroz.
		\newblock {\em Sur la d\'ecomposition r\'eelle et alg\'ebrique des syst\'emes
			d\'ependant de param\'etres}.
		\newblock PhD thesis, Universit\'e Pierre et Marie Curie - Paris VI, 2008, \url{https://tel.archives-ouvertes.fr/tel-00812436/file/these\_moroz.pdf}.
		
		\bibitem{Paszke-et-al-2019-Pytorch}
		Adam Paszke, Sam Gross, Francisco Massa, Adam Lerer, James Bradbury, Gregory
		Chanan, Trevor Killeen, Zeming Lin, Natalia Gimelshein, Luca Antiga, Alban
		Desmaison, Andreas K\"{o}pf, Edward Yang, Zach DeVito, Martin Raison, Alykhan
		Tejani, Sasank Chilamkurthy, Benoit Steiner, Lu~Fang, Junjie Bai, and Soumith
		Chintala.
		\newblock {\em PyTorch: an imperative style, high-performance deep learning
			library}.
		\newblock In Proceedings of the 33rd International Conference on Neural
		Information Processing Systems, page Article number 721. Curran Associates
		Inc., 2019, DOI:
		\href{https://doi.org/10.5555/3454287.3455008}{10.5555/3454287.3455008}.
		
		\bibitem{PdREC24}
		Lynn Pickering, Tereso {Del Rio Almajano}, Matthew England, and Kelly Cohen.
		\newblock {\em Explainable {AI} Insights for Symbolic Computation: {A} case
			study on selecting the variable ordering for cylindrical algebraic
			decomposition}.
		\newblock Journal of Symbolic Computation, 123:102276, 2024, DOI:
		\href{https://doi.org/10.1016/j.jsc.2023.102276}{10.1016/j.jsc.2023.102276}.
		
		\bibitem{Rost-Amir-2021}
		Gergely R\"{o}st and AmirHosein Sadeghimanesh.
		\newblock {\em Exotic bifurcations in three connected populations with
			{A}llee-effect}.
		\newblock International Journal of Bifurcation and Chaos, 31(13):Article number
		2150202, 2021, DOI:
		\href{https://doi.org/10.1142/S0218127421502023}{10.1142/S0218127421502023}.
		
		\bibitem{Rost-Amir-2023}
		Gergely R\"{o}st and AmirHosein Sadeghimanesh.
		\newblock {\em Unidirectional Migration of Populations with {A}llee Effect}.
		\newblock Letters in Biomathematics, 10(1):43--52, 2023, DOI:
		\href{https://doi.org/10.30707/LiB10.1.1682014077.816387}{10.30707/LiB10.1.1682014077.816387}.
		
		\bibitem{Sadeghimanesh-England}
		AmirHosein Sadeghimanesh and Mtthew England.
		\newblock {\em Resultant Tools for Parametric Polynomial Systems with
			Application to Population Models}. 2022,
		\newblock arXiv: \href{https://arxiv.org/2201.13189}{2201.13189}.
		
		\bibitem{Sadeghimanesh-Feliu-2019}
		AmirHosein Sadeghimanesh and Elisenda Feliu.
		\newblock {\em The Multistationarity Structure of Networks with Intermediates
			and a Binomial Core Network}.
		\newblock Bulletin of Mathematical Biology, 81(7):2428--2462, 2019, DOI:
		\href{https://doi.org/10.1007/s11538-019-00612-1}{10.1007/s11538-019-00612-1}.
		
		\bibitem{Sen-Bagci-Camurdan-2011}
		S.~Murat Sen, E.~Zerrin Bagci, and Mehmet~C. Camurdan.
		\newblock {\em Bistability Analysis of an Apoptosis Model in the Presence of
			Nitric Oxide}.
		\newblock Bulletin of Mathematical Biology, 73(8):1952--1968, 2011, DOI:
		\href{https://doi.org/10.1007/s11538-010-9613-5}{10.1007/s11538-010-9613-5}.
		
		\bibitem{Thomas-Kaufman-2001}
		R.~Thomas and M.~Kaufman.
		\newblock {\em Multistationarity, the basis of cell differentiation and memory.
			{I}. {S}tructural conditions of multistationarity and other nontrivial
			behavior}.
		\newblock Chaos: An Interdisciplinary Journal of Nonlinear Science,
		11(1):170--179, 2001, DOI:
		\href{https://doi.org/10.1063/1.1350439}{10.1063/1.1350439}.
		
		\bibitem{selfattention}
		Ashish Vaswani, Noam Shazeer, Niki Parmar, Jakob Uszkoreit, Llion Jones,
		Aidan~N. Gomez, \L{}ukasz Kaiser, and Illia Polosukhin.
		\newblock {\em Attention is all you need}.
		\newblock In Proceedings of the 31st International Conference on Neural
		Information Processing Systems, NIPS'17, page 6000–6010, Red Hook, NY, USA,
		2017. Curran Associates Inc, \url{https://papers.nips.cc/paper/2017}.
		
		\bibitem{GAT-paper-2018}
		Petar Veli{\v{c}}kovi{\'{c}}, Guillem Cucurull, Arantxa Casanova, Adriana
		Romero, Pietro Li{\`{o}}, and Yoshua Bengio.
		\newblock {\em Graph Attention Networks}.
		\newblock In International Conference on Learning Representations, pages 1--12,
		2018, \url{https://openreview.net/forum?id=rJXMpikCZ}.
		
		\bibitem{Wang-et-al-2019-DGL-library}
		Minjie Wang, Da~Zheng, Zihao Ye, Quan Gan, Mufei Li, Xiang Song, Jinjing Zhou,
		Chao Ma, Lingfan Yu, Yu~Gai, Tianjun Xiao, Tong He, George Karypis, Jinyang
		Li, and Zheng Zhang.
		\newblock {\em Deep Graph Library: A Graph-Centric, Highly-Performant Package
			for Graph Neural Networks}. 2019,
		\newblock arXiv: \url{https://arxiv.org/1909.01315}{1909.01315},
		
		\bibitem{Wen-et-al-2023}
		Mingjian Wen, Evan Walter~Clark Spotte-Smith, Samuel~M. Blau, Matthew~J.
		McDermott, Aditi~S. Krishnapriyan, and Kristin~A. Persson.
		\newblock {\em Chemical reaction networks and opportunities for machine
			learning}.
		\newblock Nature Computational Science, 3(1):12--24, 2023, DOI:
		\href{https://doi.org/10.1038/s43588-022-00369-z}{10.1038/s43588-022-00369-z}.
		
		\bibitem{Zhang-Ma-2012}
		Cha Zhang and Yunqian Ma.
		\newblock {\em Ensemble Machine Learning}.
		\newblock Springer New York, 2012, DOI:
		\href{https://doi.org/10.1007/978-1-4419-9326-7}{10.1007/978-1-4419-9326-7}.
		
	\end{thebibliography}
}
\end{document}